\documentclass[11pt]{article}

\usepackage[margin=1in]{geometry}
\usepackage{amsmath,amssymb,amsthm,mathtools}
\usepackage{microtype}
\usepackage{tikz}
\usepackage{booktabs}
\usepackage{array}
\usepackage{enumitem}
\usepackage{hyperref}
\usepackage{xcolor}
\usepackage{caption}
\usepackage{cite}
\usepackage{longtable}
\usepackage{listings}
\usetikzlibrary{arrows.meta,positioning,calc,decorations.pathmorphing,patterns,shapes.geometric}

\hypersetup{
  colorlinks=true,
  linkcolor=blue!50!black,
  citecolor=blue!50!black,
  urlcolor=blue!50!black
}

\newtheorem{definition}{Definition}[section]
\newtheorem{problem}[definition]{Problem}
\newtheorem{theorem}[definition]{Theorem}
\newtheorem{proposition}[definition]{Proposition}
\newtheorem{lemma}[definition]{Lemma}
\newtheorem{corollary}[definition]{Corollary}
\newtheorem{remark}[definition]{Remark}
\newtheorem{example}[definition]{Example}

\newcommand{\Z}{\mathbb{Z}}
\newcommand{\F}{\mathbb{F}}
\newcommand{\N}{\mathbb{N}}
\newcommand{\R}{\mathbb{R}}
\newcommand{\OmegaPaths}{\Omega_{n,T,q}}
\newcommand{\gammaStar}{\gamma^{\star}}
\newcommand{\eps}{\varepsilon}
\newcommand{\Prob}{\mathbb{P}}
\newcommand{\E}{\mathbb{E}}
\newcommand{\Hmin}{H_{\infty}}

\newcommand{\Enc}{\operatorname{Enc}}
\newcommand{\im}{\operatorname{im}}
\newcommand{\rank}{\operatorname{rank}}
\newcommand{\Ker}{\operatorname{ker}}
\newcommand{\Supp}{\operatorname{supp}}

\newcommand{\Adv}{\operatorname{Adv}}

\newcommand{\Setup}{\mathsf{Setup}}
\newcommand{\Sample}{\mathsf{Sample}}

\newcommand{\pp}{\mathsf{pp}}
\newcommand{\pk}{\mathsf{pk}}
\newcommand{\sk}{\mathsf{sk}}

\newcommand{\negl}{\operatorname{negl}}

\title{Exact Hidden Paths in Noisy High Dimensional Path Spaces\thanks{Mathematical and cryptographic framework draft for exact hidden path recovery. This paper proposes a candidate exact recovery relation and attack oriented research program, not a deployable cryptosystem.}}
\author{Victor Duarte Melo\\Independent Researcher\\Erd\"os number 4}
\date{\today}

\begin{document}
\maketitle

\begin{abstract}
We introduce a mathematical and cryptographic framework for studying exact path recovery in noisy high dimensional discrete path spaces inspired by the path integral formulation of quantum mechanics. In the Feynman path integral view, a transition quantity is obtained by summing contributions over all possible trajectories between boundary conditions. Standard physical methods often rely on saddlepoint, semiclassical, or variational approximations, since the goal is usually to estimate a global amplitude or dominant macroscopic behavior. This work considers a different inverse problem with cryptographic motivation: recovering one exact hidden trajectory from incomplete, noisy, projected, and aggregated path observables.

Let
\[
\gammaStar=(x_0,x_1,\ldots,x_T), \qquad x_i\in \Z_q^n,
\]
be a hidden path in a high dimensional finite space. Each transition is modified by macro increments, discrete micro perturbations, and noise,
\[
x_{i+1}=x_i+\Delta_i+\eps_i+\eta_i \pmod q,
\]
where \(\eps_i\) represents microscopic perturbations and \(\eta_i\) may be sampled from a discrete Gaussian distribution. The public data consists not of the path itself, and not of a small hash digest, but of a large vector of projected, quantized, nonlinear, or aggregated observables derived from the hidden trajectory.

We formalize the Exact Noisy Hidden Path Recovery Problem and its cryptographic search version. The objective is to reconstruct the precise path and its perturbation sequence from the public observables. The main distinction is between approximate reconstruction and exact recovery. Saddlepoint or variational methods may reveal dominant regions, coarse geometry, or average behavior, but they do not determine the exact microscopic sequence defining the hidden path. In this setting, a near solution is not a solution: even one incorrect transition or perturbation changes the recovered object.

We analyze how the admissible path space grows with dimension, path length, branching factor, perturbation entropy, and noise structure. We prove elementary information bounds showing that the effective recovery problem is limited by the amount of information contained in the published observables. Consequently, a large path space should not be compressed into a short digest if the mathematical goal is to study exact recovery. Large observable vectors are treated as meaningful public keys that preserve enough structure for verification and analysis while still obscuring the exact trajectory.

The cryptographic part of the paper develops a candidate one way relation, formal recovery games, entropy conditions, generic classical and quantum attack estimates, and a conservative roadmap toward future public key constructions. The proposed framework is not a complete encryption scheme and makes no deployment claim. Its purpose is to define a precise candidate hard problem for post quantum cryptographic investigation, together with the mathematical constraints that such a candidate must satisfy.
\end{abstract}

\tableofcontents

\section{Introduction}

The path integral formulation of quantum mechanics, introduced by Feynman \cite{feynman1948} and developed in many mathematical and physical directions \cite{feynmanhibbs1965,schulman1981,kleinert2009}, replaces the idea of a single classical trajectory by a sum over all admissible trajectories. Informally, the transition amplitude from a boundary condition \(a\) to a boundary condition \(b\) is written as
\begin{equation}
K(b,a)=\int \mathcal{D}\gamma\,\exp\left(\frac{i}{\hbar}S[\gamma]\right),
\label{eq:pathintegral}
\end{equation}
where \(S[\gamma]\) is the action of the path \(\gamma\). This expression is not merely a computational formula. It is a conceptual change: instead of asking which path is real, the formalism assigns a contribution to every path and combines the contributions by phase.

Exact path integrals are rarely available in closed form. One often studies saddlepoints, semiclassical expansions, variational approximations, perturbation series, Euclidean continuations, lattice approximations, Monte Carlo methods, or stochastic representations \cite{feynmanhibbs1965,schulman1981,kleinert2009}. These techniques are powerful because many physical questions concern global quantities. A physicist may want an amplitude, a partition function, a correlation function, or the dominant contribution in a limit. For such questions, approximation is often meaningful.

This paper studies a different mathematical question. Suppose that among a very large family of discrete paths there is a particular hidden path. The goal is not to approximate the global path sum. The goal is to recover that exact path, including microscopic perturbations and noise variables. In that setting, a saddlepoint approximation may give useful macroscopic information while still failing completely to recover the object of interest.

The guiding distinction of the paper is
\begin{equation}
\text{approximating a path sum} \quad \ne \quad \text{recovering one exact path.}
\label{eq:distinction}
\end{equation}
This distinction is simple, but it changes the nature of the inverse problem. If a candidate solution is judged by exact equality of a discrete encoding, then a path that is geometrically close may still be wrong. A single incorrect microscopic perturbation can produce a different discrete trajectory, a different observable vector, or a different final exact object. In this sense, approximate reconstruction and exact recovery become separate tasks.

We define a discrete high dimensional path model. A path is a sequence
\[
\gamma=(x_0,x_1,\ldots,x_T), \qquad x_i\in\Z_q^n,
\]
and each step is governed by a macro increment, a micro perturbation, and a noise vector. The public information is a vector of observables
\[
Y=\Phi(\gamma,\eps,\eta),
\]
where \(\Phi\) may involve projections, aggregated sums, quantization, nonlinear transformations, and randomized observation maps. The central inverse problem asks whether one can recover \((\gamma,\eps,\eta)\) exactly from \(Y\) and public parameters.

A key design choice is that the observable vector is allowed to be large. A small digest, for example a 256 bit hash, would compress the public information so aggressively that the effective inverse problem would be dominated by the information content of the digest rather than by the path structure. The framework in this paper explicitly separates a large mathematical observable vector from optional identifiers or checksums.

\subsection{Contributions}

The contributions of this paper are as follows.

\begin{enumerate}[label=(\roman*)]
\item We formulate a discrete high dimensional path space inspired by path sums and path integrals.
\item We introduce micro perturbed noisy paths, where exact recovery requires reconstructing both the trajectory and the microscopic variables that define it.
\item We define the Exact Noisy Hidden Path Recovery Problem as a mathematical inverse problem.
\item We prove elementary counting and information bounds for admissible path spaces, projected observables, compressed public data, and observable fibers.
\item We formalize the separation between approximate reconstruction and exact recovery in a discrete setting using exact distance, fiber membership, and microscopic sensitivity.
\item We discuss how dimension, branching, perturbation entropy, Gaussian noise, projection, nonlinearity, and quantization interact.
\item We propose a hierarchy of observable systems, from linear projections to nonlinear quantized path summaries, and identify structural leakage risks.
\item We formulate a cryptographic search assumption based on exact noisy hidden path recovery and define one way recovery games.
\item We analyze why large observable vectors are necessary if the public object is intended to encode more than a short digest.
\item We discuss generic classical and quantum attacks, including exhaustive recovery and Grover type search, and separate these from structural attacks such as linearization, lattice reduction, satisfiability solving, dynamic programming, and statistical leakage.
\item We outline a conservative path from mathematical relation to future cryptographic primitives, while explicitly avoiding claims that a complete public key encryption or key encapsulation mechanism has already been obtained.
\end{enumerate}

The present goal is twofold. First, we define the mathematical object and the exact recovery problem. Second, we formulate the associated cryptographic assumption carefully enough that it can be studied, attacked, and falsified.

\subsection{Relation to inverse problems and discrete noise}

The framework is connected at a high level to inverse problems, compressed sensing, noisy linear systems, lattice problems, post quantum cryptography, and discrete Gaussian techniques. Compressed sensing studies recovery of structured signals from limited measurements \cite{candes2006,donoho2006}. Lattice based mathematics and cryptography study noisy linear relations, discrete Gaussian distributions, and worst case or average case problem formulations \cite{micciancio2009,regev2005}. Modern public key cryptography often begins by identifying a search or decision problem for which no efficient classical or quantum algorithm is known \cite{goldreich2001,katzlindell2014}. The present model differs in emphasis. The unknown object is not merely a vector. It is a time ordered microscopic path with local transition variables, noise, and possibly nonlinear observables. The goal is not approximate estimation under a norm, but exact recovery of the hidden microscopic object.

\section{From path integrals to discrete path sums}

\subsection{The continuous intuition}

The formal expression \eqref{eq:pathintegral} contains a symbolic measure \(\mathcal{D}\gamma\) over paths. Although the rigorous construction of such objects depends on the setting, the guiding idea is that the transition quantity is assembled from contributions of all paths satisfying boundary conditions. In Euclidean settings one often writes a weighted expression of the form
\begin{equation}
Z=\int \mathcal{D}\gamma\,\exp(-S[\gamma]),
\label{eq:euclidean}
\end{equation}
which resembles a partition function. The difference between oscillatory and decaying weights is important in physics, but the present paper only needs the broader structural idea: a path sum is a global aggregate over an enormous family of paths.

For mathematical inverse problems it is often useful to replace a continuous path space by a finite discrete one. This eliminates analytic issues that are not central here and allows exact counting. We therefore work over finite modules and finite path families. The resulting model is not a claim about quantum dynamics. It is a discrete mathematical abstraction inspired by path integral reasoning.

\begin{figure}[t]
\centering
\begin{tikzpicture}[scale=0.95, every node/.style={font=\small}]
  \draw[->,thick] (-0.2,0) -- (5.2,0) node[right] {time};
  \draw[->,thick] (0,-1.4) -- (0,1.7) node[above] {state};
  \filldraw[blue!70!black] (0.3,-0.8) circle (2pt) node[left] {$a$};
  \filldraw[blue!70!black] (4.8,1.0) circle (2pt) node[right] {$b$};
  \draw[gray!55, thick, decorate, decoration={snake, amplitude=0.7mm, segment length=4mm}] (0.3,-0.8) .. controls (1.0,0.6) and (2.0,-1.1) .. (4.8,1.0);
  \draw[gray!55, thick, decorate, decoration={snake, amplitude=0.6mm, segment length=3mm}] (0.3,-0.8) .. controls (1.1,-1.3) and (3.0,1.5) .. (4.8,1.0);
  \draw[gray!55, thick, decorate, decoration={snake, amplitude=0.5mm, segment length=5mm}] (0.3,-0.8) .. controls (1.3,0.2) and (3.3,-0.4) .. (4.8,1.0);
  \draw[red!70!black, very thick] (0.3,-0.8) .. controls (1.0,-0.3) and (2.7,0.3) .. (4.8,1.0);
  \node[align=center] at (2.7,-1.55) {A path sum aggregates many paths.\\Exact recovery asks for one precise hidden path.};
\end{tikzpicture}
\caption{The path sum intuition. Many trajectories contribute to a global quantity, but the inverse problem studied here asks for exact recovery of one hidden trajectory.}
\label{fig:pathsum}
\end{figure}
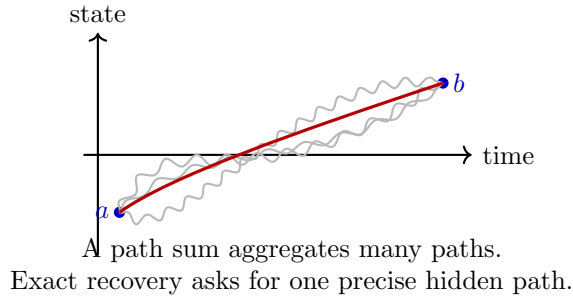

\subsection{Finite path sums}

Let \(\Omega\) be a finite set of admissible paths. A finite path sum has the form
\begin{equation}
Z=\sum_{\gamma\in\Omega} w(\gamma),
\label{eq:finitepathsum}
\end{equation}
where \(w:\Omega\to R\) is a weight taking values in a ring, field, or numerical domain. A common physical intuition is that some families of paths dominate the sum while others cancel or contribute negligibly. This motivates approximation.

The inverse problem in this paper is not to evaluate \(Z\). Instead, a hidden path \(\gammaStar\in\Omega\) is used to generate observables. The public data may resemble summaries that could have arisen from path sum calculations, but the hidden object is one precise trajectory. This shifts attention from evaluating \eqref{eq:finitepathsum} to recovering a preimage under an observation map.

\begin{definition}[Observation map]
Let \(\mathcal{X}\) be a finite set of microscopic path objects. An observation map is a function
\[
\Phi:\mathcal{X}\to \mathcal{Y},
\]
where \(\mathcal{Y}\) is the observable space. For \(Y\in\mathcal{Y}\), the fiber over \(Y\) is
\[
\Phi^{-1}(Y)=\{X\in\mathcal{X}:\Phi(X)=Y\}.
\]
\end{definition}

In the present model, an element of \(\mathcal{X}\) contains a path, perturbations, and noise. Exact recovery means finding the specific hidden element, not merely any element in a large fiber unless the problem is defined only up to equivalence.

\section{Discrete high dimensional paths}

\subsection{State space and boundary constraints}

Let \(q\ge 2\) be an integer, let \(n\ge 1\) be a dimension, and let \(T\ge 1\) be a path length. We write \(\Z_q=\Z/q\Z\). A state is a vector in \(\Z_q^n\). A discrete path of length \(T\) is a sequence of \(T+1\) states. Unless stated otherwise, \(\Z_q^n\) is used only as a finite abelian group. Whenever rank, kernel dimension, linear independence, or Gaussian elimination over the base ring is used, we explicitly assume that \(q\) is prime and identify \(\Z_q\) with the finite field \(\F_q\).

\begin{definition}[Discrete path space]
The full discrete path space is
\begin{equation}
\OmegaPaths=\{(x_0,x_1,\ldots,x_T): x_i\in\Z_q^n\}.
\end{equation}
If boundary conditions \(x_0=a\) and \(x_T=b\) are fixed, the boundary constrained path space is denoted by
\begin{equation}
\Omega_{n,T,q}(a,b)=\{\gamma\in\OmegaPaths:x_0=a,\ x_T=b\}.
\end{equation}
\end{definition}

The full path space has cardinality
\[
|\OmegaPaths|=q^{n(T+1)}.
\]
The boundary constrained path space has cardinality
\[
|\Omega_{n,T,q}(a,b)|=q^{n(T-1)},
\]
provided \(T\ge 1\). This already grows exponentially in both dimension and length.

Often one does not allow arbitrary jumps. A path may be generated by choosing increments from a set of allowed directions.

\begin{definition}[Branching model]
Let \(\mathcal{D}\subseteq \Z_q^n\) be a set of allowed macro increments, with \(|\mathcal{D}|=b\). A macro path is generated by
\begin{equation}
x_{i+1}=x_i+\Delta_i\pmod q, \qquad \Delta_i\in\mathcal{D}.
\end{equation}
The number \(b\) is called the branching factor.
\end{definition}

If no endpoint is fixed, the number of possible increment sequences is \(b^T\). If endpoints are fixed, different increment sequences may collide at the same endpoint, but the search space of possible histories remains at most \(b^T\) and can be close to that size when collisions are rare for the chosen parameters.

\begin{theorem}[Elementary growth of path histories]
\label{thm:growth}
Let \(\mathcal{D}\) be an allowed increment set of size \(b\). The number of macro increment histories of length \(T\) is exactly \(b^T\). If, in addition, each step has an independent micro perturbation chosen from a set \(\mathcal{E}\) of size \(r\), then the number of macro and micro histories is exactly \((br)^T\).
\end{theorem}

\begin{proof}
At each of the \(T\) steps there are \(b\) independent choices for \(\Delta_i\). Thus the number of macro histories is the product \(b\cdot b\cdots b=b^T\). If each step also has \(r\) choices for \(\eps_i\), then each step has \(br\) combined choices, hence \((br)^T\) histories.
\end{proof}

\begin{example}
If \(b=256\) and \(T=256\), then the number of macro histories is
\[
256^{256}=2^{2048}.
\]
If each step also has \(r=256\) micro perturbations, the number becomes
\[
(256\cdot 256)^{256}=2^{4096}.
\]
These are counting statements, not security statements. They show that exhaustive enumeration of histories becomes infeasible very quickly, but they do not rule out structural attacks.
\end{example}

\subsection{Endpoint constrained histories}

The elementary count \(b^T\) counts histories before imposing an endpoint. If the endpoint is fixed, the count depends on additive cancellations among increments. This distinction should be made explicit because endpoint constraints can reduce the number of admissible histories without making the inverse problem trivial.

Assume in this subsection that \(q\) is prime and that paths evolve over the additive group \(\F_q^n\). Let \(\mathcal{D}\subseteq \F_q^n\) be an allowed increment set. For \(a,b\in \F_q^n\), define
\[
N_T^{\mathcal{D}}(a,b)=\#\left\{(\Delta_0,\ldots,\Delta_{T-1})\in\mathcal{D}^T:\sum_{i=0}^{T-1}\Delta_i=b-a\right\}.
\]

\begin{theorem}[Endpoint constrained increment count]
\label{thm:character_count}
Let \(\widehat{\F_q^n}\) denote the additive character group of \(\F_q^n\). Then
\[
N_T^{\mathcal{D}}(a,b)=\frac{1}{q^n}\sum_{\chi\in\widehat{\F_q^n}}\chi(a-b)\left(\sum_{d\in\mathcal{D}}\chi(d)\right)^T.
\]
In particular, if \(\mathcal{D}=\F_q^n\), then \(N_T^{\mathcal{D}}(a,b)=q^{n(T-1)}\) for every \(a,b\).
\end{theorem}

\begin{proof}
Use the character orthogonality identity on the finite abelian group \(\F_q^n\):
\[
\frac{1}{q^n}\sum_{\chi\in\widehat{\F_q^n}}\chi(u)=
\begin{cases}
1,&u=0,\\
0,&u\ne 0.
\end{cases}
\]
Therefore
\[
N_T^{\mathcal{D}}(a,b)
=\sum_{\Delta_0,\ldots,\Delta_{T-1}\in\mathcal{D}}
\frac{1}{q^n}\sum_{\chi}\chi\left(a-b+\sum_{i=0}^{T-1}\Delta_i\right).
\]
Interchanging the sums gives
\[
N_T^{\mathcal{D}}(a,b)=\frac{1}{q^n}\sum_{\chi}\chi(a-b)
\prod_{i=0}^{T-1}\left(\sum_{d\in\mathcal{D}}\chi(d)\right),
\]
which is the claimed formula. If \(\mathcal{D}=\F_q^n\), then \(\sum_{d\in\F_q^n}\chi(d)=0\) for every nontrivial character and equals \(q^n\) for the trivial character, so the count is \(q^{-n}(q^n)^T=q^{n(T-1)}\).
\end{proof}

\begin{remark}
Theorem~\ref{thm:character_count} is useful because it separates the size of the history space from the size of the endpoint constrained history space. A path generator can have \(b^T\) possible histories, while the number leading to a particular endpoint depends on the additive spectrum of \(\mathcal{D}\). This is one place where structural leakage can enter: highly symmetric increment sets may create counts that are easier to analyze than intended.
\end{remark}

\begin{corollary}[Endpoint count with independent transition alphabets]
\label{cor:endpoint_combined_alphabet}
Assume \(q\) is prime and let \(\mathcal{A}\subseteq\F_q^n\) be the set of all possible total transition increments
\[
u_i=\Delta_i+\eps_i+\eta_i.
\]
If the effective transition choice at each step ranges over \(\mathcal{A}\), then the number of transition histories leading from \(a\) to \(b\) in \(T\) steps is
\[
\frac{1}{q^n}\sum_{\chi\in\widehat{\F_q^n}}\chi(a-b)\left(\sum_{u\in\mathcal{A}}m(u)\chi(u)\right)^T,
\]
where \(m(u)\) is the number of triples \((\Delta,\eps,\eta)\) producing the effective increment \(u\). If every effective increment has the same multiplicity and \(\mathcal{A}=\F_q^n\), the count is that multiplicity to the power \(T\) times \(q^{n(T-1)}\).
\end{corollary}

\begin{proof}
Repeat the character orthogonality proof of Theorem~\ref{thm:character_count}, replacing the unweighted increment set by the weighted multiset of effective increments. The multiplicity function \(m(u)\) records how many microscopic transition triples produce the same total increment. In the uniform full alphabet case, only the trivial character survives, giving the stated expression.
\end{proof}

\begin{proposition}[Microscopic decomposition multiplicity]
\label{prop:transition_decomposition_multiplicity}
Let \(\mathcal{D},\mathcal{E},\mathcal{N}\subseteq \Z_q^n\) be finite transition alphabets and define the effective transition map
\[
\mu:\mathcal{D}\times\mathcal{E}\times\mathcal{N}\to \Z_q^n,
\qquad
\mu(\Delta,\eps,\eta)=\Delta+\eps+\eta \pmod q.
\]
For \(u\in\Z_q^n\), let
\[
m(u)=|\mu^{-1}(u)|.
\]
If an effective increment sequence \((u_0,\ldots,u_{T-1})\) is fixed, then the number of microscopic transition histories \((\Delta_i,\eps_i,\eta_i)_{i=0}^{T-1}\) that realize it is
\[
\prod_{i=0}^{T-1}m(u_i).
\]
Consequently, observing the effective state path alone does not identify the microscopic transition history unless \(m(u_i)=1\) for every realized increment.
\end{proposition}

\begin{proof}
For each time index \(i\), the number of triples producing the fixed effective increment \(u_i\) is exactly \(m(u_i)\) by definition. The choices at different time indices are independent once the effective increment sequence is fixed, so the multiplication rule gives \(\prod_i m(u_i)\). If some realized increment has multiplicity greater than one, at least two different microscopic transition histories produce the same effective state transition sequence, so the microscopic history is not identified by the state path alone.
\end{proof}

\subsection{Path encodings and exact equality}

A path object must be encoded in a finite representation before exact equality can be discussed. Let
\[
X=(\gamma,\Delta,\eps,\eta)
\]
be a complete microscopic object. An encoding function
\[
\Enc:\mathcal{X}\to\{0,1\}^L
\]
assigns a bit string to each object. In mathematical sections, exact recovery means recovering \(X\). In computational sections, exact recovery means recovering \(\Enc(X)\).

\begin{definition}[Exact recovery]
Given public parameters \(P\), an observable \(Y=\Phi_P(X^{\star})\), and a candidate \(\widetilde X\), the candidate is an exact object recovery if
\[
\widetilde X=X^{\star}.
\]
If the recovery target is the encoded object, the candidate is an exact encoded recovery if
\[
\Enc(\widetilde X)=\Enc(X^{\star}).
\]
These two notions coincide whenever \(\Enc\) is injective on the admissible domain. If \(\Enc\) is not injective, encoded recovery is a quotient recovery problem rather than recovery of the full microscopic object. A candidate that is close under a metric but not equal to the required target is called an approximate reconstruction.
\end{definition}

The distinction seems elementary, but it is central. A path may be close under Euclidean, Hamming, or edit distance and still fail exact equality.

\section{High dimension and projected views}

High dimension plays two roles. First, it increases the state space. Second, it makes partial observations less informative. A low dimensional projection of a high dimensional vector discards many degrees of freedom. This section records exact versions of that principle over finite fields.

When linear algebra is used, we assume that \(q\) is prime, so \(\Z_q=\F_q\) is a finite field.

\begin{definition}[Coordinate projection]
Let \(P:\F_q^n\to \F_q^k\) be a linear map with rank \(k\le n\). For a path \(\gamma=(x_0,\ldots,x_T)\), the projected path is
\[
P\gamma=(Px_0,Px_1,\ldots,Px_T).
\]
\end{definition}

\begin{theorem}[Projection ambiguity]
\label{thm:projection}
Let \(P:\F_q^n\to \F_q^k\) be linear of rank \(k\). For every \(y\in\im(P)\), the set \(P^{-1}(y)\) has exactly \(q^{n-k}\) elements. Consequently, for a projected path \((y_0,\ldots,y_T)\in(\im(P))^{T+1}\), the number of full paths mapping to it is exactly
\[
q^{(n-k)(T+1)}.
\]
If the full endpoints are fixed to \(x_0=a\) and \(x_T=b\), and if the projected endpoints satisfy \(y_0=P a\) and \(y_T=P b\), then the number of full paths with these fixed endpoints and projected intermediate values \((y_1,\ldots,y_{T-1})\) is exactly
\[
q^{(n-k)(T-1)}.
\]
If either endpoint compatibility condition fails, the number is zero.
\end{theorem}

\begin{proof}
Since \(P\) has rank \(k\), the rank nullity theorem gives \(\dim \Ker(P)=n-k\). For any \(y\in\im(P)\), choose one \(x_y\) with \(Px_y=y\). Then every solution is of the form \(x_y+z\) with \(z\in\Ker(P)\), and there are \(|\Ker(P)|=q^{n-k}\) such vectors. Applying this independently to each of the \(T+1\) projected states gives \(q^{(n-k)(T+1)}\) full paths. If the full endpoints are fixed, then the first and last states no longer contribute free kernel choices. The remaining \(T-1\) intermediate states contribute independently, giving \(q^{(n-k)(T-1)}\). If \(y_0\ne P a\) or \(y_T\ne P b\), no compatible full path exists.
\end{proof}

\begin{figure}[t]
\centering
\begin{tikzpicture}[scale=1.0, every node/.style={font=\small}]
  \draw[rounded corners, thick, fill=blue!4] (-3.4,-1.6) rectangle (0.8,1.6);
  \node at (-1.3,1.35) {high dimensional state};
  \foreach \x/\y in {-2.8/-1.0,-2.4/0.4,-1.9/-0.1,-1.3/0.9,-0.8/-0.6,-0.2/0.7,0.3/-0.2} {
    \filldraw[blue!70!black] (\x,\y) circle (1.4pt);
  }
  \draw[red!70!black, very thick] (-2.8,-1.0) -- (-2.4,0.4) -- (-1.9,-0.1) -- (-1.3,0.9) -- (-0.8,-0.6) -- (-0.2,0.7) -- (0.3,-0.2);
  \draw[->,thick] (1.0,0) -- (2.2,0) node[midway,above] {$P$};
  \draw[rounded corners, thick, fill=green!4] (2.4,-1.6) rectangle (5.4,1.6);
  \node at (3.9,1.35) {low dimensional view};
  \draw[gray!50, thick] (2.8,-1.0) -- (5.0,1.0);
  \foreach \t in {0,1,2,3,4,5,6} {
    \pgfmathsetmacro{\xx}{2.9+0.32*\t}
    \pgfmathsetmacro{\yy}{-0.8+0.23*\t}
    \filldraw[red!70!black] (\xx,\yy) circle (1.5pt);
  }
  \node[align=center] at (1.0,-2.05) {Projection can preserve coarse motion while discarding exact degrees of freedom.};
\end{tikzpicture}
\caption{A low dimensional observable can preserve a visible trend while losing a large number of microscopic degrees of freedom. Theorem~\ref{thm:projection} gives an exact finite field count.}
\label{fig:projection}
\end{figure}

Theorem~\ref{thm:projection} is not a claim that projections always make recovery difficult. It states that projection alone creates information loss. Whether a recovery algorithm can use additional structure depends on the observable map. This is why later sections emphasize the design of observables rather than dimension alone.

\subsection{Distance concentration intuition}

In high dimensional probability, many geometric quantities concentrate around typical values \cite{vershynin2018}. Although our model is finite and discrete, the same intuition is useful. If points are sampled from \(\Z_q^n\) with large \(n\), many projected or aggregated statistics can become similar for many unrelated paths. This can help hide microscopic choices, but it can also make the inverse problem underdetermined.

The mathematical challenge is therefore not simply to increase \(n\). The challenge is to choose observable maps that create a meaningful exact recovery problem. Too much information reveals the path. Too little information leaves many equally valid paths. The regime of interest lies between these extremes.

\subsection{A finite concentration statement}

The following elementary theorem gives a precise version of the high dimensional intuition for Hamming distance over \(\F_q^n\). It is not a hardness theorem. It says that many random high dimensional objects look typical under coarse distance statistics.

\begin{theorem}[Hamming distance concentration]
\label{thm:hamming_concentration}
Let \(U,V\) be independent uniform random variables in \(\F_q^n\). Let
\[
D_H(U,V)=|\{j:U_j\ne V_j\}|.
\]
Then
\[
D_H(U,V)\sim \operatorname{Binomial}\left(n,1-\frac{1}{q}\right),
\]
so
\[
\E[D_H(U,V)]=n\left(1-\frac{1}{q}\right).
\]
Moreover, for every \(t>0\),
\[
\Prob\left[\left|D_H(U,V)-n\left(1-\frac{1}{q}\right)\right|\ge t\right]
\le 2\exp\left(-\frac{2t^2}{n}\right).
\]
\end{theorem}

\begin{proof}
For each coordinate \(j\), define \(I_j=1\) if \(U_j\ne V_j\), and \(I_j=0\) otherwise. Since \(U_j\) and \(V_j\) are independent uniform elements of \(\F_q\), we have \(\Prob[I_j=1]=1-1/q\). The variables \(I_1,\ldots,I_n\) are independent Bernoulli variables, and \(D_H(U,V)=\sum_{j=1}^n I_j\). The expectation and binomial law follow immediately. The tail bound is Hoeffding's inequality for independent variables in \([0,1]\).
\end{proof}

\begin{theorem}[Hamming sphere size]
\label{thm:hamming_sphere}
Fix \(u\in\F_q^n\). For \(0\le k\le n\), the number of vectors \(v\in\F_q^n\) satisfying \(D_H(u,v)=k\) is
\[
\binom{n}{k}(q-1)^k.
\]
Consequently, for every integer \(r\) with \(0\le r\le n\), the number of vectors within radius \(r\) of \(u\) is
\[
\sum_{k=0}^{r}\binom{n}{k}(q-1)^k.
\]
\end{theorem}

\begin{proof}
To construct a vector \(v\) at distance \(k\) from \(u\), first choose the \(k\) coordinates in which \(v\) differs from \(u\). This gives \(\binom{n}{k}\) choices. On each chosen coordinate there are exactly \(q-1\) values different from the corresponding coordinate of \(u\). On all other coordinates the value is forced to equal \(u\). Multiplying gives \(\binom{n}{k}(q-1)^k\). Summing over \(0\le k\le r\) gives the ball size.
\end{proof}

\begin{corollary}[Coarse distance degeneracy]
A single Hamming distance value from a fixed reference state cannot identify a microscopic state unless the corresponding sphere has size one. In particular, for any \(1\le k\le n\), the number of states with that distance is \(\binom{n}{k}(q-1)^k\), which is exponential in \(n\) whenever \(k=\alpha n\) with fixed \(0<\alpha<1\) and the usual nondegenerate edge cases are excluded.
\end{corollary}

\begin{proof}
The first statement follows because all states on the same Hamming sphere produce the same distance observable. The exact number of such states is Theorem~\ref{thm:hamming_sphere}. For \(k=\alpha n\) with fixed \(0<\alpha<1\), Stirling's formula gives exponential growth in \(n\) of \(\binom{n}{\alpha n}\). Multiplication by \((q-1)^{\alpha n}\) preserves exponential growth. The only degenerate boundary cases are radii such as \(k=0\), and for \(q=2\) also \(k=n\), where the corresponding sphere may have size one.
\end{proof}

\section{Micro perturbations and noisy transitions}

\subsection{Transition model}

The basic hidden path model uses a transition rule
\begin{equation}
x_{i+1}=x_i+\Delta_i+\eps_i+\eta_i \pmod q,
\label{eq:transition}
\end{equation}
where \(\Delta_i\in\mathcal{D}\) is a macro increment, \(\eps_i\in\mathcal{E}\) is a micro perturbation, and \(\eta_i\in\Z_q^n\) is a noise vector. The sequence
\[
(\Delta_0,\ldots,\Delta_{T-1})
\]
controls the coarse trajectory, while the sequence
\[
(\eps_0,\ldots,\eps_{T-1})
\]
controls the exact microscopic deviations. The noise sequence
\[
(\eta_0,\ldots,\eta_{T-1})
\]
can be fixed as part of the hidden object or can be used as part of the observable generation process.

\begin{definition}[Complete microscopic path object]
A complete microscopic path object is a tuple
\[
X=(x_0,\Delta,\eps,\eta)
\]
where \(x_0\in\Z_q^n\), \(\Delta\in\mathcal{D}^T\), \(\eps\in\mathcal{E}^T\), and \(\eta\in\mathcal{N}^T\). The path \(\gamma(X)=(x_0,\ldots,x_T)\) is obtained by iterating \eqref{eq:transition}.
\end{definition}

This definition makes the hidden object larger than the visible path. Two different triples \((\Delta,\eps,\eta)\) can sometimes generate the same state sequence. Depending on the problem version, exact recovery may require the path states only or the full transition decomposition. Both versions are meaningful, but they differ.

\begin{definition}[State recovery and microscopic recovery]
State recovery asks for the exact state sequence \(\gammaStar=(x_0,\ldots,x_T)\). Microscopic recovery asks for the full tuple \((x_0,\Delta,\eps,\eta)\). Microscopic recovery is at least as strong as state recovery.
\end{definition}

\subsection{Discrete Gaussian noise}

A discrete Gaussian distribution is a standard way to place concentrated but nonuniform noise on a lattice or finite module. For \(\sigma>0\), the integer discrete Gaussian on \(\Z\) assigns probability proportional to
\[
\rho_\sigma(z)=\exp\left(-\frac{\pi z^2}{\sigma^2}\right).
\]
A vector noise distribution can be obtained by sampling coordinates independently and reducing modulo \(q\), or by sampling from a spherical discrete Gaussian on \(\Z^n\) before reduction. Discrete Gaussian distributions are widely used in lattice based mathematics and cryptography \cite{micciancio2009}.

In this paper the role of Gaussian noise is conceptual. It increases microscopic ambiguity and separates coarse reconstruction from exact recovery. If a recovery method estimates the macro path but fails to recover the exact noise and perturbation variables, it does not solve the microscopic recovery problem.

There is a small but important technical point. A discrete Gaussian on \(\Z\), reduced modulo \(q\), may have full support on \(\Z_q\), even when most of the probability mass is concentrated near zero before reduction. Therefore counting by support size can be misleading unless the distribution is explicitly truncated. When the distribution is not truncated, entropy, tail bounds, or effective support should be used instead of literal support size.

\begin{proposition}[Canonical recovery of bounded integer noise]
\label{prop:bounded_noise_aliasing}
Let \(B\in\N\) and suppose an integer noise coordinate \(e\) is sampled from \([-B,B]\cap\Z\) and then reduced modulo \(q\). The reduction map is injective on \([-B,B]\cap\Z\) if and only if \(q>2B\). In this injective case, the residue class determines the original integer value. If \(q\le 2B\), there exist two distinct integers in the interval with the same residue modulo \(q\), so recovering a residue does not necessarily recover the original integer sample.
\end{proposition}

\begin{proof}
If \(e_1,e_2\in[-B,B]\cap\Z\) have the same residue modulo \(q\), then \(q\) divides \(e_1-e_2\). Also \(|e_1-e_2|\le 2B\). When \(q>2B\), the only multiple of \(q\) in \([-2B,2B]\) is zero, hence \(e_1=e_2\). Conversely, if \(q\le 2B\), then \(-B\) and \(-B+q\) both belong to \([-B,B]\cap\Z\), are distinct, and are congruent modulo \(q\). Thus the reduction map is not injective.
\end{proof}

\begin{definition}[Noisy micro path distribution]
Fix \(x_0\), allowed increments \(\mathcal{D}\), micro perturbations \(\mathcal{E}\), and a noise distribution \(D_\sigma\) on \(\Z_q^n\). A noisy micro path is generated by independently sampling
\[
\Delta_i\leftarrow\mathcal{D},\qquad \eps_i\leftarrow\mathcal{E},\qquad \eta_i\leftarrow D_\sigma
\]
and applying \eqref{eq:transition} for \(i=0,\ldots,T-1\).
\end{definition}

\begin{figure}[t]
\centering
\begin{tikzpicture}[scale=0.95, every node/.style={font=\small}]
  \draw[->,thick] (-0.2,0) -- (7.1,0) node[right] {time};
  \foreach \i/\x/\y in {0/0.4/0.0,1/1.4/0.45,2/2.5/-0.15,3/3.6/0.65,4/4.7/0.05,5/5.9/0.9} {
    \filldraw[red!75!black] (\x,\y) circle (2pt) node[above] {$x_\i$};
  }
  \draw[red!70!black, very thick] (0.4,0.0) -- (1.4,0.45) -- (2.5,-0.15) -- (3.6,0.65) -- (4.7,0.05) -- (5.9,0.9);
  \foreach \x/\y/\dx/\dy in {0.4/0.0/0.55/0.10,1.4/0.45/0.48/-0.25,2.5/-0.15/0.58/0.30,3.6/0.65/0.52/-0.18,4.7/0.05/0.55/0.25} {
    \draw[blue!70!black,->,thick] (\x,\y) -- ++(\dx,\dy);
    \draw[green!50!black,->,thick] (\x+\dx,\y+\dy) -- ++(0.18,0.12);
    \draw[purple!70!black,->,thick] (\x+\dx+0.18,\y+\dy+0.12) -- ++(0.12,-0.08);
  }
  \node[blue!70!black] at (1.2,-0.9) {$\Delta_i$ macro};
  \node[green!50!black] at (3.2,-0.9) {$\eps_i$ micro};
  \node[purple!70!black] at (5.1,-0.9) {$\eta_i$ noise};
  \node[align=center] at (3.5,-1.45) {The visible step is a sum of coarse motion, microscopic perturbation, and noise.};
\end{tikzpicture}
\caption{A transition decomposed into macro increment, micro perturbation, and noise. Approximate methods may estimate the coarse trend while missing microscopic terms.}
\label{fig:micro}
\end{figure}

\section{Observable systems}

\subsection{Public observables}

The observable vector is the central mathematical object. It should not be treated as a short digest. Instead, it is a structured vector derived from the hidden path object.

\begin{definition}[Observable system]
An observable system is a tuple
\[
\mathcal{O}=(\mathcal{X},\mathcal{Y},\Phi),
\]
where \(\mathcal{X}\) is a microscopic path object space, \(\mathcal{Y}\) is an observable space, and \(\Phi:\mathcal{X}\to\mathcal{Y}\) is an observation map. If \(\mathcal{Y}=\mathcal{Y}_1\times\cdots\times\mathcal{Y}_m\), then
\[
\Phi(X)=(Q_1(X),\ldots,Q_m(X))
\]
for component observables \(Q_j:\mathcal{X}\to\mathcal{Y}_j\).
\end{definition}

Examples of component observables include:

\begin{enumerate}[label=(\alph*)]
\item projected endpoint sums,
\[
Q_j(X)=\sum_{i=0}^{T} \langle a_{j,i},x_i\rangle \pmod q;
\]
\item transition energy summaries,
\[
Q_j(X)=\sum_{i=0}^{T-1} \psi_j(x_i,x_{i+1},i) \pmod q;
\]
\item quantized real valued summaries,
\[
Q_j(X)=\left\lfloor \frac{1}{\tau_j}\sum_{i=0}^{T-1} f_j(x_i,x_{i+1},\eps_i,\eta_i) \right\rceil;
\]
\item nonlinear local features,
\[
Q_j(X)=\sum_{i=0}^{T-1} \chi_j(i)\,g_j(x_i,x_{i+1},\eps_i) \pmod q.
\]
\end{enumerate}

The model intentionally permits many choices. The paper does not claim that every observable system yields a difficult recovery problem. Some choices are bad. For example, if \(\Phi\) publishes all states \((x_0,\ldots,x_T)\), state recovery is trivial. If \(\Phi\) publishes too little, exact recovery may be information theoretically impossible. The interesting regime is between these extremes.

\begin{figure}[t]
\centering
\begin{tikzpicture}[scale=0.82, transform shape, node distance=0.65cm, every node/.style={font=\small}]
  \node[draw,rounded corners,fill=blue!5,minimum width=2.1cm,minimum height=0.8cm] (hidden) {hidden object $X$};
  \node[draw,rounded corners,fill=green!5,minimum width=2.1cm,minimum height=0.8cm,right=of hidden] (features) {local features};
  \node[draw,rounded corners,fill=yellow!10,minimum width=2.1cm,minimum height=0.8cm,right=of features] (agg) {aggregation};
  \node[draw,rounded corners,fill=orange!10,minimum width=2.1cm,minimum height=0.8cm,right=of agg] (quant) {quantization};
  \node[draw,rounded corners,fill=red!5,minimum width=2.2cm,minimum height=0.8cm,right=of quant] (Y) {observable vector $Y$};
  \draw[->,thick] (hidden) -- (features);
  \draw[->,thick] (features) -- (agg);
  \draw[->,thick] (agg) -- (quant);
  \draw[->,thick] (quant) -- (Y);
  \node[align=center] at (5.4,-1.1) {A large observable vector is the mathematical public object.\\A short checksum may identify it, but should not replace it.};
\end{tikzpicture}
\caption{Observable generation pipeline. The observables are structured mathematical data, not merely a short digest.}
\label{fig:pipeline}
\end{figure}
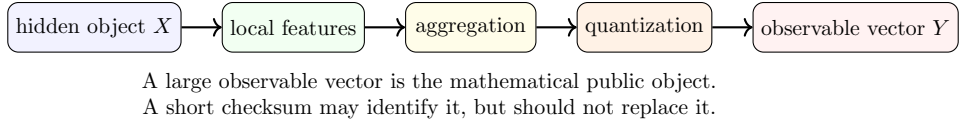

\subsection{Linear observable systems}

Linear systems are easy to analyze and useful as baselines. Assume \(q\) is prime and write a microscopic object as a vector \(v\in\F_q^N\). A linear observable system has the form
\[
Y=Av,
\]
where \(A\in\F_q^{m\times N}\). Linear systems are not necessarily suitable for hard recovery, because rank reveals their information content exactly.

\begin{proposition}[Linear fiber size]
\label{prop:linear_fiber}
Let \(A\in\F_q^{m\times N}\) have rank \(r_A\). If \(Y\in\im(A)\), then
the fiber \(\{v\in\F_q^N:Av=Y\}\) has size
\[
q^{N-r_A}.
\]
\end{proposition}

\begin{proof}
The solution set is either empty or an affine coset of \(\Ker(A)\). Since \(\dim\Ker(A)=N-r_A\), each nonempty fiber has \(q^{N-r_A}\) elements.
\end{proof}

This proposition shows both the usefulness and weakness of purely linear observables. They are analyzable, but if the hidden object is not constrained by additional nonlinear structure, exact recovery is impossible unless \(r_A=N\). If \(r_A=N\), recovery is linear algebra.

\subsection{Nonlinear and quantized observables}

Nonlinear observables can prevent direct linear inversion, while quantization can destroy fine continuity. A simple class is
\begin{equation}
Q_j(X)=\left\lfloor \frac{F_j(X)+e_j}{\tau_j}\right\rceil,
\label{eq:quantized}
\end{equation}
where \(F_j\) is a real or integer valued feature, \(e_j\) is observation noise, and \(\tau_j\) is a quantization scale.

Quantization has two competing effects. It hides small differences and may increase ambiguity. It can also merge many objects into the same observable bin, which may make exact recovery information theoretically impossible. Thus quantization must be studied as part of the mathematical problem, not merely as an implementation detail.

\begin{remark}
The phrase noisy hidden path should not be interpreted as a license to add arbitrary noise. If too much noise is added to the public observables, the hidden path becomes unrecoverable even for ideal algorithms. The goal is not to destroy all information, but to create a problem where exact microscopic recovery is meaningfully harder than approximate macroscopic reconstruction.
\end{remark}

\section{The exact noisy hidden path recovery problem}

We now define the central problem.

\begin{definition}[Parameter set]
A parameter set is
\[
P=(q,n,T,\mathcal{D},\mathcal{E},D_\sigma,\Phi,\Enc),
\]
where \(q,n,T\) define the state space, \(\mathcal{D}\) defines macro increments, \(\mathcal{E}\) defines micro perturbations, \(D_\sigma\) defines noise, \(\Phi\) defines observables, and \(\Enc\) defines exact encoding.
\end{definition}

\begin{problem}[Exact Noisy Hidden Path Recovery]
\label{prob:ENHPR}
Given public parameters \(P\) and an observable vector
\[
Y=\Phi(X^{\star})
\]
for a hidden microscopic path object \(X^{\star}\in\mathcal{X}_P\), recover \(X^{\star}\) exactly. If an encoded version is chosen as the formal target, recover \(\Enc(X^{\star})\) exactly; this is equivalent to object recovery only when \(\Enc\) is injective on the relevant admissible set.
\end{problem}

Several variants are useful.

\begin{problem}[Exact State Path Recovery]
Given \(Y=\Phi(X^{\star})\), recover the state sequence \(\gamma(X^{\star})\).
\end{problem}

\begin{problem}[Exact Transition Recovery]
Given \(Y=\Phi(X^{\star})\), recover \((\Delta^{\star},\eps^{\star},\eta^{\star})\) and enough boundary information to reconstruct the path.
\end{problem}

\begin{problem}[Fiber Search]
Given \(Y\), find any \(X\in\Phi^{-1}(Y)\).
\end{problem}

The exact hidden object problem is stronger than fiber search when multiple objects share the same observables. If the intended equivalence relation treats all objects in the same fiber as equivalent, then fiber search is the correct problem. If the hidden object itself is the target, then collisions matter.

\begin{definition}[Collision class]
For \(X\in\mathcal{X}\), its collision class under \(\Phi\) is
\[
[X]_{\Phi}=\{X'\in\mathcal{X}:\Phi(X')=\Phi(X)\}.
\]
The observable system is injective on a subset \(S\subseteq\mathcal{X}\) if \([X]_{\Phi}\cap S=\{X\}\) for every \(X\in S\).
\end{definition}

If injectivity holds on the generated subset, exact recovery is information theoretically possible. If not, exact recovery of the specific object is impossible from observables alone unless additional information is available.

\section{Identifiability before computation}

Before discussing algorithms, one must ask whether the hidden object is identifiable from the observable vector at all. This is a purely mathematical issue. Computational hardness is meaningful only after identifiability or an explicit equivalence relation has been specified.

\begin{definition}[Generation support]
For a parameter set \(P\), let \(S_P\subseteq\mathcal{X}_P\) be the set of microscopic objects that can be generated with nonzero probability by the stated sampling procedure. We call \(S_P\) the generation support.
\end{definition}

\begin{theorem}[Identifiability criterion]
\label{thm:identifiability}
Exact recovery of every generated object from \(Y=\Phi(X)\) is information theoretically possible if and only if \(\Phi\) is injective on \(S_P\).
\end{theorem}

\begin{proof}
If \(\Phi\) is injective on \(S_P\), then every observable \(Y\in\Phi(S_P)\) has a unique preimage in \(S_P\). An ideal unbounded procedure can recover that unique preimage by exhaustive search over \(S_P\). Conversely, if \(\Phi\) is not injective on \(S_P\), then there exist distinct \(X_1,X_2\in S_P\) with \(\Phi(X_1)=\Phi(X_2)\). The same public observable is consistent with two different hidden objects, so no procedure using only \(Y\) and public parameters can recover both correctly for all generated objects.
\end{proof}

\begin{theorem}[Best possible success inside a fiber]
\label{thm:fiber_success}
Let \(X\) be uniformly distributed on a finite set \(S\subseteq\mathcal{X}\), and let \(Y=\Phi(X)\). For any recovery algorithm \(A\),
\[
\Prob[A(Y)=X]\le \E\left[\frac{1}{|\Phi^{-1}(Y)\cap S|}\right]
\]
when \(A\) is required to succeed for the exact hidden object and has no side information beyond \(Y\). Equality is achieved by choosing uniformly from the fiber whenever the fiber can be enumerated.
\end{theorem}

\begin{proof}
Condition on the event \(Y=y\). Under the uniform prior on \(S\), the conditional distribution of \(X\) is uniform on \(\Phi^{-1}(y)\cap S\). The best possible guess succeeds with probability \(1/|\Phi^{-1}(y)\cap S|\). Averaging over \(Y\) gives the bound, and uniform selection from the fiber achieves it.
\end{proof}

\begin{corollary}[Optimal success for uniform planted recovery]
\label{cor:uniform_planted_success}
Let \(X\) be uniform on a finite support \(S\) and let \(Y=\Phi(X)\). The optimal probability of recovering the planted object exactly from \(Y\) is
\[
\frac{|\Phi(S)|}{|S|}.
\]
In particular, large fibers reduce planted exact recovery even when every fiber can be enumerated perfectly.
\end{corollary}

\begin{proof}
By Theorem~\ref{thm:fiber_success}, the optimal success probability is
\[
\E\left[\frac{1}{|\Phi^{-1}(Y)\cap S|}\right].
\]
Writing \(k_y=|\Phi^{-1}(y)\cap S|\), this equals
\[
\sum_{y\in\Phi(S)} \frac{k_y}{|S|}\frac{1}{k_y}=\frac{|\Phi(S)|}{|S|}.
\]
\end{proof}

\begin{proposition}[Bayes optimal exact recovery]
\label{prop:bayes_recovery}
Let \(X\) be any random hidden object on a finite set and let \(Y=\Phi(X)\). Among all recovery rules that observe only \(Y\), the maximum success probability is
\[
\sum_y \Prob[Y=y]\max_x \Prob[X=x\mid Y=y].
\]
In particular, if \(X\) is uniform on each conditional fiber, this reduces to the fiber expression in Theorem~\ref{thm:fiber_success}.
\end{proposition}

\begin{proof}
Condition on \(Y=y\). Any recovery rule must output some candidate \(a(y)\). Its conditional success probability is \(\Prob[X=a(y)\mid Y=y]\), which is maximized by choosing a value of \(x\) with maximum posterior probability. Averaging the optimal conditional success probability over all \(y\) gives the formula. If the posterior distribution is uniform on a fiber of size \(k_y\), the maximum is \(1/k_y\), giving Theorem~\ref{thm:fiber_success}.
\end{proof}

\begin{remark}
The theorem is not about efficient algorithms. It says that if the observable map creates large fibers and the target is the exact hidden object, then the obstacle is not merely computational. The object is not identified by the data. In such a regime, one must either publish more observables, change the target to fiber recovery, or add an equivalence relation.
\end{remark}

\section{Information content and the no short digest principle}

The size of the hidden path space does not by itself determine the difficulty of recovery. The observable vector determines how much information is exposed. This section states basic bounds.

\begin{theorem}[Average fiber lower bound]
\label{thm:fiber_lower}
Let \(\Phi:\mathcal{X}\to\mathcal{Y}\) be any map between finite sets. Then there exists an observable \(Y\in\mathcal{Y}\) such that
\[
|\Phi^{-1}(Y)|\ge \frac{|\mathcal{X}|}{|\mathcal{Y}|}.
\]
Moreover, if \(X\) is uniform on \(\mathcal{X}\), the average fiber size seen by the image distribution is
\[
\sum_{Y\in\mathcal{Y}} \Prob[\Phi(X)=Y]\ |\Phi^{-1}(Y)|
=\frac{1}{|\mathcal{X}|}\sum_{Y\in\mathcal{Y}} |\Phi^{-1}(Y)|^2.
\]
\end{theorem}

\begin{proof}
The first statement follows from
\[
|\mathcal{X}|=\sum_{Y\in\mathcal{Y}} |\Phi^{-1}(Y)|
\]
and the pigeonhole principle. The second statement follows because \(\Prob[\Phi(X)=Y]=|\Phi^{-1}(Y)|/|\mathcal{X}|\).
\end{proof}

\begin{proposition}[Collision probability and visible fiber size]
\label{prop:collision_fiber_identity}
Let \(S\) be a finite set with \(|S|=N\), let \(X,X'\) be independent uniform random variables on \(S\), and let \(Y=\Phi(X)\). Then
\[
\Prob[\Phi(X)=\Phi(X')]
=\frac{1}{N^2}\sum_y |\Phi^{-1}(y)\cap S|^2,
\]
and
\[
\E[|\Phi^{-1}(Y)\cap S|]
= N\,\Prob[\Phi(X)=\Phi(X')].
\]
Thus collision probability is exactly the normalized average fiber size seen by a random public value.
\end{proposition}

\begin{proof}
Write \(k_y=|\Phi^{-1}(y)\cap S|\). Since \(X\) and \(X'\) are independent and uniform on \(S\),
\[
\Prob[\Phi(X)=\Phi(X')]
=\sum_y \Prob[\Phi(X)=y]\Prob[\Phi(X')=y]
=\sum_y \left(\frac{k_y}{N}\right)^2.
\]
Also, because \(Y=\Phi(X)\),
\[
\E[|\Phi^{-1}(Y)\cap S|]
=\sum_y \Prob[Y=y] k_y
=\sum_y \frac{k_y}{N}k_y
=\frac{1}{N}\sum_y k_y^2.
\]
Combining the two identities gives the result.
\end{proof}

\begin{corollary}[Short digest limitation]
\label{cor:digest}
If \(|\mathcal{X}|=2^N\) and \(\Phi:\mathcal{X}\to\{0,1\}^\ell\), then some fiber has size at least \(2^{N-\ell}\). In particular, if \(N=4096\) and \(\ell=256\), some fiber has size at least \(2^{3840}\).
\end{corollary}

\begin{proof}
Apply Theorem~\ref{thm:fiber_lower} with \(|\mathcal{Y}|=2^\ell\).
\end{proof}

\begin{proposition}[Post processing cannot improve identifiability]
\label{prop:postprocessing_fibers}
Let \(\Phi:S\to\mathcal{Y}\) be an observation map and let \(h:\mathcal{Y}\to\mathcal{Z}\) be any deterministic post processing map. Define \(\Psi=h\circ\Phi\). For every \(z\in\mathcal{Z}\),
\[
\Psi^{-1}(z)=\bigcup_{y\in h^{-1}(z)}\Phi^{-1}(y).
\]
In particular, post processing can only merge observable fibers; it cannot split them. If \(h\) is not injective on \(\Phi(S)\), then at least one visible class of \(\Psi\) is the union of two or more distinct \(\Phi\)-fibers.
\end{proposition}

\begin{proof}
For \(x\in S\), the equality \(\Psi(x)=z\) is equivalent to \(h(\Phi(x))=z\), which is equivalent to \(\Phi(x)\in h^{-1}(z)\). This proves the displayed union identity. The remaining statements follow immediately from the union identity. Deterministic post processing discards information unless it is injective on the actually occurring public values.
\end{proof}

\begin{proposition}[Random observable collision estimate]
\label{prop:random_collision}
Let \(|\mathcal{X}|=N\) and \(|\mathcal{Y}|=M\). If \(\Phi:\mathcal{X}\to\mathcal{Y}\) is chosen uniformly at random among all functions, then the expected number of unordered colliding pairs is
\[
\E\left[\#\{\{x,x'\}:x\ne x',\ \Phi(x)=\Phi(x')\}\right]=\binom{N}{2}\frac{1}{M}.
\]
\end{proposition}

\begin{proof}
For each unordered pair \(\{x,x'\}\), define an indicator variable \(I_{x,x'}\) that equals one when \(\Phi(x)=\Phi(x')\). Since the two images are independent uniform elements of \(\mathcal{Y}\), \(\Prob[I_{x,x'}=1]=1/M\). Linearity of expectation over the \(\binom{N}{2}\) pairs gives the result.
\end{proof}

Corollary~\ref{cor:digest} is one of the central lessons of the framework. A huge internal path space should not be collapsed into a short output if the paper aims to study recovery from structured observables. A short digest may be useful as an identifier or checksum, but it should not replace the observable vector.

\begin{proposition}[Information upper bound]
Let \(X\) be any random hidden object and let \(Y=\Phi(X)\) take values in a finite observable space \(\mathcal{Y}\). Then
\[
I(X;Y)\le H(Y)\le \log_2 |\mathcal{Y}|.
\]
If \(Y\) is represented by \(L\) bits, then \(I(X;Y)\le L\).
\end{proposition}

\begin{proof}
This is the standard information theoretic inequality \(I(X;Y)\le H(Y)\), together with \(H(Y)\le \log_2|\mathcal{Y}|\). If \(Y\) is represented by \(L\) bits, then \(|\mathcal{Y}|\le 2^L\).
\end{proof}

Thus the mathematical public object should have enough information content to make the inverse problem meaningful. Large public observables are not automatically bad. They may be necessary to avoid reducing the problem to a short digest preimage question.

\begin{proposition}[Conditional entropy after a deterministic observation]
Let \(X\) be uniform on a finite set \(\mathcal{X}\), and let \(Y=\Phi(X)\). Then
\[
H(X\mid Y)=\sum_{y\in\Phi(\mathcal{X})}\frac{|\Phi^{-1}(y)|}{|\mathcal{X}|}\log_2 |\Phi^{-1}(y)|.
\]
In particular, if all nonempty fibers have size at least \(M\), then \(H(X\mid Y)\ge \log_2 M\).
\end{proposition}

\begin{proof}
For a fixed observable \(y\), the conditional distribution of \(X\) given \(Y=y\) is uniform over \(\Phi^{-1}(y)\). Therefore
\[
H(X\mid Y=y)=\log_2 |\Phi^{-1}(y)|.
\]
Averaging this identity over \(Y\) gives the formula. The lower bound follows immediately if every nonempty fiber has size at least \(M\).
\end{proof}

\begin{theorem}[Fano type obstruction to exact recovery]
\label{thm:fano_obstruction}
This standard information theoretic bound is included in a form adapted to exact recovery \cite{cover2006}. Let \(X\) be uniform on a finite set \(S\) with \(|S|=N\ge 2\), let \(Y\) be any observation of \(X\), and let \(\widehat X=A(Y)\) be any recovery rule. If
\[
P_e=\Prob[\widehat X\ne X],
\]
then
\[
H(X\mid Y)\le h_2(P_e)+P_e\log_2(N-1),
\]
where \(h_2(p)=-p\log_2 p-(1-p)\log_2(1-p)\) is the binary entropy function, with the convention \(0\log 0=0\). Consequently, the weaker but often useful bound
\[
P_e\ge \frac{H(X\mid Y)-1}{\log_2 N}
\]
holds whenever the numerator is positive.
\end{theorem}

\begin{proof}
Let \(E\) be the error indicator, equal to one when \(\widehat X\ne X\). Since \(\widehat X\) is a function of \(Y\), conditioning on \(Y\) also reveals \(\widehat X\). We use
\[
H(X\mid Y)\le H(E,X\mid Y)=H(E\mid Y)+H(X\mid E,Y).
\]
The first term satisfies \(H(E\mid Y)\le H(E)=h_2(P_e)\). If \(E=0\), then \(X=\widehat X\) is determined by \(Y\), so \(H(X\mid E=0,Y)=0\). If \(E=1\), then \(X\) can take at most \(N-1\) values different from \(\widehat X\), so \(H(X\mid E=1,Y)\le \log_2(N-1)\). Averaging over \(E\) gives Fano's inequality. Since \(h_2(P_e)\le 1\) and \(\log_2(N-1)\le \log_2 N\), we obtain
\[
H(X\mid Y)\le 1+P_e\log_2 N,
\]
which rearranges to the displayed lower bound.
\end{proof}

This theorem is useful for the present framework because it turns large conditional uncertainty into a lower bound on exact recovery error. It is not a computational hardness theorem. It is an information theoretic obstruction that applies even to unbounded recovery rules.

\begin{corollary}[Necessary observable length for injectivity]
If \(\Phi:S\to\{0,1\}^L\) is injective on a finite set \(S\), then
\[
L\ge \lceil\log_2 |S|\rceil.
\]
\end{corollary}

\begin{proof}
Injectivity gives \(|S|\le 2^L\). Taking base two logarithms gives the claim.
\end{proof}

\subsection{Observable length table}

Table~\ref{tab:observable_size} gives simple examples. The numbers are not recommended parameters. They illustrate the relationship between observable count, bits per observable, and public information length.

\begin{table}[t]
\centering
\begin{tabular}{rrrr}
\toprule
observables \(m\) & bits each \(\ell\) & total bits & total bytes \\
\midrule
256 & 16 & 4096 & 512 \\
512 & 16 & 8192 & 1024 \\
1024 & 16 & 16384 & 2048 \\
1024 & 32 & 32768 & 4096 \\
4096 & 16 & 65536 & 8192 \\
8192 & 16 & 131072 & 16384 \\
\bottomrule
\end{tabular}
\caption{Large observable vectors can carry far more information than a short digest while still being practical for mathematical experiments.}
\label{tab:observable_size}
\end{table}

\section{Exact recovery versus approximate reconstruction}

\subsection{Metrics on path objects}

To compare exact and approximate recovery, define distances on paths. For state paths, a natural Hamming type distance is
\[
d_H(\gamma,\gamma')=|\{i:x_i\ne x'_i\}|.
\]
For microscopic objects, one can combine state, increment, perturbation, and noise discrepancies:
\[
d_{\mathcal{X}}(X,X')=d_H(\gamma,\gamma')+d_H(\Delta,\Delta')+d_H(\eps,\eps')+d_H(\eta,\eta').
\]
Exact recovery is the case \(d_{\mathcal{X}}(X,X')=0\).

\begin{definition}[Approximation radius]
For \(r\ge 0\), the ball of radius \(r\) around \(X\) is
\[
B_r(X)=\{X'\in\mathcal{X}:d_{\mathcal{X}}(X,X')\le r\}.
\]
A recovery algorithm is an \(r\) approximate reconstruction algorithm if it outputs \(X'\in B_r(X^{\star})\).
\end{definition}

Approximate reconstruction can be meaningful for estimation, visualization, or physical interpretation. It is not the same as exact recovery.

\begin{proposition}[Near is not exact]
For any discrete metric \(d\) on \(\mathcal{X}\), exact recovery is equivalent to recovery within radius \(0\). If an exact verification rule accepts only \(X=X^{\star}\), then every \(X'\ne X^{\star}\), regardless of how small \(d(X,X')\) is, fails exact verification.
\end{proposition}

\begin{proof}
By the definition of a metric, \(d(X,X^{\star})=0\) if and only if \(X=X^{\star}\). Exact verification accepts exactly this case.
\end{proof}

\subsection{Why saddlepoint intuition does not solve exact recovery}

A saddlepoint method usually identifies stationary or dominant configurations of an action. Suppose a discrete action \(S\) is defined on \(\mathcal{X}\), and a method finds
\[
\widehat X\in\operatorname*{argmin}_{X\in\mathcal{X}} S(X)
\]
or a nearby approximate minimizer. If the hidden path \(X^{\star}\) is not uniquely determined by being a global minimizer, then the saddlepoint does not recover \(X^{\star}\). Even if \(X^{\star}\) lies near a dominant region, its microscopic perturbations may remain hidden.

\begin{theorem}[Dominant region ambiguity]
\label{thm:dominant_region}
Let \(\mathcal{X}\) be finite and let \(R\subseteq\mathcal{X}\) be a dominant region identified by an approximation method. If \(|R|>1\) and the hidden object \(X^{\star}\) is only known to lie in \(R\), then exact recovery of \(X^{\star}\) from membership in \(R\) alone is impossible.
\end{theorem}

\begin{proof}
If \(|R|>1\), there exist distinct \(X_1,X_2\in R\). The information that the hidden object lies in \(R\) is consistent with both \(X_1\) and \(X_2\). No deterministic or randomized procedure using only this information can identify the hidden object with probability \(1\) for all possible hidden objects in \(R\).
\end{proof}

Theorem~\ref{thm:dominant_region} is elementary, but it captures the central idea. Approximation can reduce a huge space to a smaller region, yet exact recovery remains unresolved if the region still contains many microscopic variants.

\begin{corollary}[Residual uncertainty after localization]
Suppose an approximation method outputs a region \(R(Y)\subseteq\mathcal{X}\) known to contain \(X^{\star}\), and suppose that conditional on this information the hidden object is uniform over \(R(Y)\). Then the optimal exact recovery probability from this localized information alone is
\[
\frac{1}{|R(Y)|}.
\]
Thus reducing a search space from \(|\mathcal{X}|\) to \(|R(Y)|\) is useful, but it is not exact recovery unless \(|R(Y)|=1\).
\end{corollary}

\begin{proof}
Under the stated conditional uniform distribution, every element of \(R(Y)\) is equally likely. The best possible exact guess has success probability \(1/|R(Y)|\), and no procedure using only the region membership information can do better.
\end{proof}

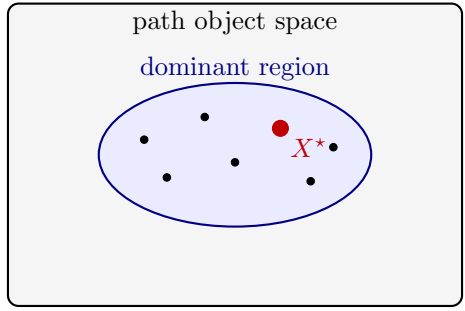
\begin{figure}[t]
\centering
\begin{tikzpicture}[scale=1.0, every node/.style={font=\small}]
  \draw[rounded corners,thick,fill=gray!8] (-3,-2) rectangle (3,2);
  \node at (0,1.75) {path object space};
  \draw[fill=blue!8,draw=blue!50!black,thick] (0,0) ellipse (1.8 and 0.95);
  \node[blue!50!black] at (0,1.15) {dominant region};
  \foreach \x/\y in {-1.2/0.2,-0.9/-0.3,-0.4/0.5,0.0/-0.1,0.6/0.35,1.0/-0.35,1.3/0.1} {
    \filldraw[black] (\x,\y) circle (1.4pt);
  }
  \filldraw[red!75!black] (0.6,0.35) circle (3pt) node[below right] {$X^\star$};
  \node[align=center] at (0,-2.45) {Approximation may locate a region. Exact recovery must identify one microscopic point.};
\end{tikzpicture}
\caption{A dominant region may contain many microscopic path objects. Locating the region does not imply exact recovery.}
\label{fig:region}
\end{figure}

\section{Counting with dimension, branching, and noise}

This section collects more explicit growth estimates. Let \(b=|\mathcal{D}|\), \(r=|\mathcal{E}|\), and let \(s=|\Supp(D_\sigma)|\) be the support size of a truncated noise distribution. If the noise is not truncated, one works with effective support or entropy instead.

\begin{proposition}[Microscopic history count]
\label{prop:micro_count}
Assume that each step independently chooses \(\Delta_i\in\mathcal{D}\), \(\eps_i\in\mathcal{E}\), and \(\eta_i\in\mathcal{N}\), with sizes \(b,r,s\). Then the number of microscopic histories of length \(T\) is
\[
(brs)^T.
\]
If \(x_0\) is also variable over \(\Z_q^n\), the number is
\[
q^n(brs)^T.
\]
\end{proposition}

\begin{proof}
For each step there are \(b r s\) choices. Multiplying over \(T\) steps gives \((brs)^T\). If \(x_0\) is variable, multiply by \(q^n\).
\end{proof}

When noise has a distribution rather than a finite uniform support, entropy is a better measure.

\begin{definition}[Min entropy]
For a discrete random variable \(Z\), the min entropy is
\[
\Hmin(Z)=-\log_2\max_z \Prob[Z=z].
\]
\end{definition}

\begin{proposition}[Min entropy of independent choices]
If \(Z_1,
\ldots,Z_T\) are independent, then
\[
\Hmin(Z_1,\ldots,Z_T)=\sum_{i=1}^T \Hmin(Z_i).
\]
In particular, if each step contributes at least \(h\) bits of min entropy, the sequence contributes at least \(Th\) bits.
\end{proposition}

\begin{proof}
Independence gives
\[
\max_{z_1,\ldots,z_T}\Prob[Z_1=z_1,\ldots,Z_T=z_T]
=\prod_{i=1}^T \max_{z_i}\Prob[Z_i=z_i].
\]
Taking \(-\log_2\) gives the result.
\end{proof}

\begin{table}[t]
\centering
\begin{tabular}{rrrrr}
\toprule
branch \(b\) & micro \(r\) & noise \(s\) & length \(T\) & histories \\
\midrule
256 & 1 & 1 & 128 & \(2^{1024}\) \\
256 & 1 & 1 & 256 & \(2^{2048}\) \\
256 & 256 & 1 & 256 & \(2^{4096}\) \\
256 & 256 & 16 & 256 & \(2^{5120}\) \\
1024 & 256 & 16 & 256 & \(2^{5632}\) \\
\bottomrule
\end{tabular}
\caption{Illustrative microscopic history counts. These numbers measure search space size, not proven hardness.}
\label{tab:history_counts}
\end{table}

\subsection{State space versus history space}

The number of histories can exceed the number of possible state paths if different transition decompositions lead to the same states. Conversely, the state path space can be huge even when the allowed transitions are restricted. It is therefore useful to distinguish:

\begin{enumerate}[label=(\alph*)]
\item state path space: possible sequences \((x_0,\ldots,x_T)\);
\item macro history space: possible increment sequences \(\Delta\);
\item microscopic history space: possible \((\Delta,\eps,\eta)\);
\item observable image: possible values \(Y=\Phi(X)\).
\end{enumerate}

The exact recovery problem should specify which of these is the target. A paper that mixes them risks overstating the result.

\section{Quantization and Gaussian ambiguity}

\subsection{Quantized observations}

Let \(F:\mathcal{X}\to\R^m\) be a feature map. A quantized observable has the form
\[
Y_j=\left\lfloor \frac{F_j(X)}{\tau_j}\right\rceil.
\]
The quantization cell of \(Y\) is
\[
C_Y=\{z\in\R^m: \lfloor z_j/\tau_j\rceil=Y_j\text{ for all }j\}.
\]
Exact recovery from \(Y\) requires distinguishing all hidden objects whose feature vectors lie in the same cell.

\begin{proposition}[Quantization fiber]
Let \(F:\mathcal{X}\to\R^m\) and let \(Q\) be a componentwise quantizer. The fiber of the quantized observable \(Y=Q(F(X))\) is
\[
(Q\circ F)^{-1}(Y)=\{X\in\mathcal{X}:F(X)\in C_Y\}.
\]
Thus quantization merges all hidden objects whose feature vectors lie in the same cell.
\end{proposition}

\begin{proof}
This follows directly from the definition of preimage under composition.
\end{proof}

Quantization is useful because it destroys fragile real valued precision. It is dangerous because it creates collisions. The mathematical study must quantify both effects.

\subsection{Gaussian noise inside the path}

If Gaussian noise is part of the hidden transition, then the hidden object contains random variables that must be recovered exactly under microscopic recovery. If Gaussian noise is added only to the public observable, then exact recovery may become impossible unless there is redundancy or error correction.

\begin{definition}[Internal noise and observation noise]
Internal noise is noise \(\eta_i\) used in the transition rule \eqref{eq:transition} and included in the hidden object. Observation noise is noise \(e_j\) added after features are computed, for example
\[
Y_j=Q_j(X)+e_j.
\]
\end{definition}

The two forms should not be confused. Internal noise increases the hidden microscopic description. Observation noise corrupts public measurements. In a later cryptographic construction, observation noise may require reconciliation. In this mathematical paper, it is sufficient to distinguish the two roles.

\begin{proposition}[Observation noise can destroy exact identifiability]
\label{prop:obs_noise_identifiability}
Let \(F:S\to A\) be a deterministic feature map on a finite hidden object set \(S\), and let the public observation be
\[
Y=F(X)+E
\]
in an abelian group \(A\), where \(E\) is independent observation noise with support \(\mathcal{B}\). If there exist distinct \(X_1,X_2\in S\) such that
\[
(F(X_1)+\mathcal{B})\cap(F(X_2)+\mathcal{B})\ne\varnothing,
\]
then exact recovery of \(X\) from \(Y\) is not identifiable on \(S\) without recovering or otherwise accounting for \(E\).
\end{proposition}

\begin{proof}
The intersection condition means that there exist \(e_1,e_2\in\mathcal{B}\) such that \(F(X_1)+e_1=F(X_2)+e_2=y\). Thus the same public observation \(y\) can arise from two different hidden objects with admissible observation noise values. Hence \(Y=y\) alone does not determine whether the hidden object was \(X_1\) or \(X_2\).
\end{proof}

\begin{remark}
This is why the present model prefers internal noise when discussing exact microscopic recovery. If noise is added only after observation, exact recovery of the original object requires redundancy, side information, an error correction mechanism, or a weaker recovery target.
\end{remark}

\begin{figure}[t]
\centering
\begin{tikzpicture}[node distance=1.05cm, every node/.style={font=\small}]
  \node[draw,rounded corners,fill=blue!5] (macro) {$\Delta$};
  \node[draw,rounded corners,fill=green!5,right=of macro] (micro) {$\eps$};
  \node[draw,rounded corners,fill=purple!5,right=of micro] (noise) {$\eta$};
  \node[draw,rounded corners,fill=red!5,right=of noise] (path) {$\gamma$};
  \node[draw,rounded corners,fill=yellow!10,right=of path] (obs) {$Y$};
  \draw[->,thick] (macro) -- (micro);
  \draw[->,thick] (micro) -- (noise);
  \draw[->,thick] (noise) -- (path);
  \draw[->,thick] (path) -- (obs);
  \node[align=center] at (3.8,-1.0) {Internal noise belongs to the hidden microscopic path.\\Observation noise corrupts the public vector.};
\end{tikzpicture}
\caption{Two roles of noise. The present framework mainly uses internal noise to increase microscopic ambiguity.}
\label{fig:noise_roles}
\end{figure}

\section{Observable design and structural leakage}

A large path space is not enough. Observable design can accidentally reveal the path. This section identifies simple leakage patterns.

\subsection{Telescoping leakage}

Some aggregate observables collapse to endpoint information. For example,
\[
\sum_{i=0}^{T-1}(x_{i+1}-x_i)=x_T-x_0.
\]
If an observable is only a telescoping sum, it carries little microscopic information.

\begin{proposition}[Telescoping loss]
Let \(Q(\gamma)=\sum_{i=0}^{T-1}(x_{i+1}-x_i)\) over an abelian group. Then
\[
Q(\gamma)=x_T-x_0.
\]
Consequently, \(Q\) is independent of all intermediate states.
\end{proposition}

\begin{proof}
The sum expands as
\[
(x_1-x_0)+(x_2-x_1)+\cdots+(x_T-x_{T-1})=x_T-x_0.
\]
All intermediate terms cancel.
\end{proof}

This is a warning. Aggregates should be checked for hidden cancellations.

\subsection{Linear leakage}

If all observables are linear in a vectorized hidden object, then the problem reduces to linear algebra plus any constraints on the object. This may be acceptable for a baseline, but it should not be mistaken for a nonlinear recovery problem.

\subsection{Overdetermined leakage}

If \(\Phi\) publishes enough independent constraints to become injective and efficiently invertible, recovery is easy. Injectivity is not automatically bad for a mathematical inverse problem, but efficient invertibility is fatal to any future hardness interpretation.

\subsection{Underdetermined collapse}

If \(\Phi\) publishes too little information, many hidden paths remain consistent with \(Y\). Then exact recovery of the specific hidden object is information theoretically impossible. This is distinct from computational difficulty.

\begin{remark}[The middle regime]
The target regime is not maximum secrecy and not full revelation. It is a structured middle regime where public observables are large enough to define a meaningful exact inverse problem, but not so direct that recovery becomes trivial.
\end{remark}

\section{A hierarchy of recovery notions}

Because exact recovery can be too strong in the presence of collisions, it is helpful to define a hierarchy.

\begin{definition}[Recovery levels]
Let \(X^{\star}\) be hidden and \(Y=\Phi(X^{\star})\).
\begin{enumerate}[label=(\roman*)]
\item Coarse recovery outputs a macroscopic descriptor \(M(X^{\star})\).
\item Approximate recovery outputs \(\widetilde X\) with \(d(\widetilde X,X^{\star})\le r\) for some \(r>0\).
\item Fiber recovery outputs some \(X\in\Phi^{-1}(Y)\).
\item Exact state recovery outputs \(\gamma(X^{\star})\).
\item Exact microscopic recovery outputs \(X^{\star}\).
\end{enumerate}
\end{definition}

The model in this paper focuses on exact microscopic recovery, while acknowledging that some applications may only need weaker notions.

\begin{theorem}[Implication chain]
Assume the hidden microscopic object determines the state path and the state path determines every coarse descriptor. Then exact microscopic recovery implies exact state recovery, and exact state recovery implies coarse recovery. The converses do not hold in general.
\end{theorem}

\begin{proof}
If \(X^{\star}\) is known, then \(
\gamma(X^{\star})\) can be computed, so exact microscopic recovery implies exact state recovery. If the state path determines a coarse descriptor, then exact state recovery implies coarse recovery. The converses fail whenever two microscopic objects generate the same state path, or two distinct state paths share the same coarse descriptor.
\end{proof}

\begin{figure}[t]
\centering
\begin{tikzpicture}[node distance=0.9cm, every node/.style={font=\small}]
  \node[draw,rounded corners,fill=red!5,minimum width=3.2cm] (micro) {exact microscopic recovery};
  \node[draw,rounded corners,fill=orange!10,minimum width=3.2cm,below=of micro] (state) {exact state recovery};
  \node[draw,rounded corners,fill=yellow!15,minimum width=3.2cm,below=of state] (fiber) {fiber recovery};
  \node[draw,rounded corners,fill=green!8,minimum width=3.2cm,below=of fiber] (approx) {approximate recovery};
  \node[draw,rounded corners,fill=blue!6,minimum width=3.2cm,below=of approx] (coarse) {coarse descriptor};
  \draw[->,thick] (micro) -- (state);
  \draw[->,thick] (state) -- (coarse);
  \draw[dashed,->,thick] (fiber) -- (coarse);
  \draw[dashed,->,thick] (approx) -- (coarse);
  \node[align=center] at (5.0,-2.0) {The arrows depend on the chosen model.\\Exact microscopic recovery is the strongest notion here.};
\end{tikzpicture}
\caption{A hierarchy of recovery notions. Approximate reconstruction and exact microscopic recovery should not be conflated.}
\label{fig:hierarchy}
\end{figure}
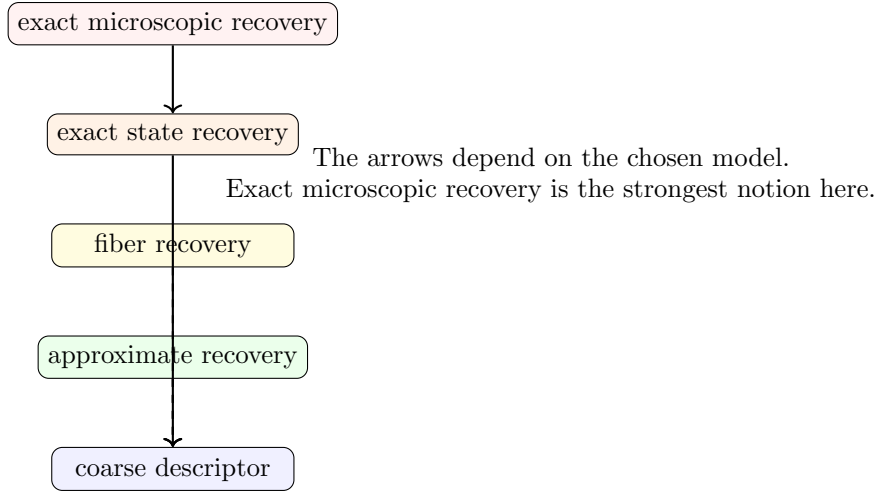

\section{Toy examples}

Toy examples clarify the difference between global summaries and exact microscopic paths.

\subsection{A one dimensional example}

Let \(q=101\), \(n=1\), \(T=4\), \(\mathcal{D}=\{-1,1\}\), and \(\mathcal{E}=\{-1,0,1\}\). Ignore noise. A microscopic history has \((2\cdot 3)^4=1296\) possible transition descriptions. If the only observable is the endpoint \(x_4\), then many histories collide.

Even if an approximate method correctly infers that the path tends upward, it does not determine the exact sequence of signs and micro perturbations. The path is short enough to enumerate, but it illustrates the conceptual gap.

\subsection{A high dimensional example}

Let \(q=257\), \(n=128\), \(T=256\), \(|\mathcal{D}|=256\), \(|\mathcal{E}|=256\), and suppose internal noise has effective support \(s=16\) per step. The microscopic history count is
\[
(256\cdot256\cdot16)^{256}=2^{5120}.
\]
If the public observable vector has \(8192\) bits, the observable space is at most \(2^{8192}\), which is large enough that the pigeonhole lower bound alone does not force collisions. But this does not imply recovery is easy or hard. It means that information content is not obviously insufficient. The actual difficulty depends on \(\Phi\).

\subsection{A bad observable and a better observable}

A bad observable is
\[
Y=x_T-x_0.
\]
It ignores the internal path. A better experimental observable could combine multiple local nonlinear summaries:
\[
Y_j=\sum_{i=0}^{T-1}\left(\langle a_{j,i},x_i\rangle^2+c_{j,i}\langle b_{j,i},x_{i+1}\rangle+\langle d_{j,i},\eps_i\rangle\right)\pmod q.
\]
This is not proposed as a final hard problem. It is a useful research object because it avoids immediate telescoping and is not purely a single endpoint measurement.

\section{Experimental mathematical methodology}

Although this paper is theoretical, the framework naturally suggests experiments. A research prototype can generate hidden paths, compute observables, and test recovery attempts. The following methodology is mathematical rather than cryptographic.

\begin{enumerate}[label=(\roman*)]
\item Generate path objects with fixed \((q,n,T,b,r,\sigma)\).
\item Compute multiple observable families \(Y=\Phi(X)\).
\item Measure observable collision rates on sampled subsets.
\item Attempt simple recovery methods such as linear solving, local search, dynamic programming, integer programming, and gradient heuristics when applicable.
\item Measure whether methods recover coarse structure, approximate paths, fibers, or exact microscopic objects.
\item Increase dimension, length, perturbation entropy, and observable size to study scaling.
\end{enumerate}

\begin{lstlisting}[caption={A minimal experimental loop in pseudocode.}]
for params in parameter_grid:
    X = sample_hidden_path(params)
    Y = compute_observables(X, params)
    for solver in solvers:
        X_guess = solver(Y, params)
        report coarse_score(X_guess, X)
        report distance(X_guess, X)
        report exact_success(X_guess == X)
        report fiber_success(Phi(X_guess) == Y)
\end{lstlisting}

The important point is to report different success notions separately. A solver that finds a visually similar path has not solved exact microscopic recovery.

\section{Potential asymptotic questions}

The framework suggests several asymptotic questions. Let a family of parameter sets \(P_\lambda\) be indexed by a size parameter \(\lambda\). Let \(\mathcal{X}_\lambda\) be the hidden object space and \(\Phi_\lambda\) the observable map.

\begin{problem}[Asymptotic exact recovery]
For which families \((P_\lambda)
\) does there exist an algorithm that recovers \(X^{\star}\) from \(Y=\Phi_\lambda(X^{\star})\) in time polynomial in \(\lambda\), with high probability over the generation of \(X^{\star}\)?
\end{problem}

\begin{problem}[Average fiber growth]
How does the expected size of \([X^{\star}]_{\Phi_\lambda}\) behave as \(\lambda\to\infty\)? Is it typically \(1\), polynomial, exponential, or concentrated around another scale?
\end{problem}

\begin{problem}[Approximation gap]
Can one construct families where coarse descriptors are recoverable in polynomial time, but exact microscopic recovery appears to require exponential time?
\end{problem}

The last problem captures the informal goal of separating saddlepoint style approximation from exact recovery.

\section{Cryptographic interpretation}

The previous sections define a mathematical inverse problem. This section reformulates it as a candidate cryptographic assumption. The purpose is not to present a complete public key encryption scheme. The purpose is to specify a search relation, a public instance distribution, and basic security games that can be attacked in later work.

A cryptographic interpretation requires more discipline than a mathematical metaphor. It is not enough to say that there are many paths. One must specify how instances are generated, what the public information is, what the secret witness is, what it means to break the system, and which algorithms are allowed to attack it.

\subsection{Public parameters, public keys, and witnesses}

Let \(\lambda\) be a security parameter. A parameter generator produces
\[
\pp_\lambda=(q,n,T,\mathcal{D},\mathcal{E},D_\sigma,\Phi,\tau,\mathcal{A}),
\]
where \(q\) is the modulus, \(n\) is the dimension, \(T\) is the path length, \(\mathcal{D}\) is the macro increment alphabet, \(\mathcal{E}\) is the micro perturbation alphabet, \(D_\sigma\) is the noise distribution, \(\Phi\) is the observable map, \(\tau\) is a quantization parameter if needed, and \(\mathcal{A}\) contains auxiliary public sampling rules.

A secret microscopic object is
\[
X^{\star}=(\gamma^{\star},\Delta^{\star},\eps^{\star},\eta^{\star}),
\]
where
\[
\gamma^{\star}=(x_0,\ldots,x_T),\qquad x_i\in\Z_q^n.
\]
The public object is
\[
\pk=Y=\Phi(X^{\star}).
\]
Thus the public key is not a short digest. It is a structured observable vector. In a later implementation it may be serialized, compressed losslessly, or accompanied by a checksum, but the mathematical object is the observable vector itself.

\begin{definition}[Witness relation]
For fixed public parameters \(\pp\), define the relation
\[
\mathcal{R}_{\pp}=\{(Y,X):\Phi(X)=Y\}.
\]
An element \(X\) is a witness for \(Y\) if \((Y,X)\in\mathcal{R}_{\pp}\).
\end{definition}

\begin{definition}[Exact witness]
If the instance generator samples \(X^{\star}\) and publishes \(Y=\Phi(X^{\star})\), then \(X^{\star}\) is the planted exact witness. Exact recovery asks for \(X^{\star}\), not merely for an arbitrary member of \(\Phi^{-1}(Y)\), unless the problem is explicitly quotient by an equivalence relation.
\end{definition}

This distinction matters. If many witnesses are accepted as equivalent, then the cryptographic problem is not planted recovery but relation finding. Relation finding may still be useful, but it must be named differently.

\begin{proposition}[Public relations do not select planted witnesses]
\label{prop:public_relation_no_selector}
Fix public parameters \(\pp\) and a public value \(Y\). Suppose the fiber
\[
F_Y=\{X:\Phi(X)=Y\}
\]
has size greater than one. Any test whose input is only \((\pp,Y,X)\) and whose acceptance condition is the public relation \(\Phi(X)=Y\) accepts every element of \(F_Y\). Therefore planted exact recovery from a non singleton fiber requires additional structure beyond public relation verification, such as injectivity on the generation support, a canonical representative rule, a private selector, or a key derivation rule that does not make arbitrary witnesses useful.
\end{proposition}

\begin{proof}
For every \(X\in F_Y\), the equality \(\Phi(X)=Y\) holds by definition. Hence a verifier that checks only the public relation accepts all elements of the fiber. If \(|F_Y|>1\), this test cannot distinguish the planted element from another valid element. The listed alternatives are precisely ways to change the mathematical target or add information that is not contained in the bare relation test.
\end{proof}

\subsection{The Exact Noisy Hidden Path one way assumption}

We now define a cryptographic search problem associated with the mathematical framework.

\begin{problem}[Exact Noisy Hidden Path Search, cryptographic form]
Let \(\mathcal{I}_\lambda\) be an instance distribution generated as follows:
\[
\pp_\lambda\leftarrow \Setup(1^\lambda),
\]
\[
X^{\star}\leftarrow \Sample(\pp_\lambda),
\]
\[
Y=\Phi_{\pp_\lambda}(X^{\star}).
\]
Given \((\pp_\lambda,Y)\), output the planted witness \(X^{\star}\).
\end{problem}

\begin{definition}[One way recovery game]
Let \(\mathcal{A}\) be an adversary. The one way recovery experiment is:
\[
\pp_\lambda\leftarrow \Setup(1^\lambda),\qquad X^{\star}\leftarrow \Sample(\pp_\lambda),\qquad Y=\Phi_{\pp_\lambda}(X^{\star}),
\]
\[
\widehat{X}\leftarrow \mathcal{A}(\pp_\lambda,Y).
\]
The adversary wins if
\[
\widehat{X}=X^{\star}.
\]
The advantage is
\[
\Adv^{\mathrm{ow}}_{\mathcal{A}}(\lambda)=
\Prob[\mathcal{A}(\pp_\lambda,Y)=X^{\star}].
\]
The family is one way for planted exact recovery if for every probabilistic polynomial time adversary \(\mathcal{A}\),
\[
\Adv^{\mathrm{ow}}_{\mathcal{A}}(\lambda)\le \negl(\lambda).
\]
\end{definition}

\begin{remark}[Why this is stronger than finding a collision]
If \(\Phi^{-1}(Y)\) is large, an adversary may find some \(X\ne X^{\star}\) with \(\Phi(X)=Y\). This does not win the planted recovery game. However, for many cryptographic applications, any valid witness may be enough to break the construction. Therefore future work must decide whether the desired assumption is planted exact recovery, arbitrary witness finding, or recovery modulo equivalence.
\end{remark}

\begin{definition}[Arbitrary witness game]
The arbitrary witness experiment is identical, except that the adversary wins if
\[
\Phi_{\pp_\lambda}(\widehat{X})=Y.
\]
Its advantage is denoted by \(\Adv^{\mathrm{rel}}_{\mathcal{A}}(\lambda)\).
\end{definition}

The arbitrary witness game is usually the more conservative cryptographic target. If a primitive is broken by any witness in the fiber, then the hardness of planted recovery is insufficient. A secure design should either make fibers essentially singleton on the generated support, define a canonical quotient whose representatives all derive the same intended value, or make alternative witnesses useless for deriving the secret quantity. Proposition~\ref{prop:public_relation_no_selector} is the formal reason this choice cannot be postponed.

\subsection{Identification, equivalence, and canonical witnesses}

The previous point leads to a necessary formal choice.

\begin{definition}[Recovery equivalence]
Let \(\sim\) be an equivalence relation on a generation support \(S\). Exact recovery modulo \(\sim\) asks for any \(\widehat{X}\) satisfying \(\widehat{X}\sim X^{\star}\). There are two distinct compatibility notions:
\begin{enumerate}[label=(\alph*)]
\item \(\sim\) is observable compatible if \(X\sim X'\) implies \(\Phi(X)=\Phi(X')\). Then \(\Phi\) is constant on equivalence classes.
\item \(\Phi\) is class identifying if \(\Phi(X)=\Phi(X')\) implies \(X\sim X'\). Then each observable fiber lies inside one equivalence class.
\end{enumerate}
Modulo recovery from \(Y=\Phi(X^{\star})\) requires the second condition. If both conditions hold, the nonempty observable fibers are exactly the equivalence classes.
\end{definition}

\begin{definition}[Canonical recovery]
A canonical recovery rule is a deterministic function
\[
\mathsf{Can}:\mathcal{X}\to\mathcal{X}
\]
such that \(\mathsf{Can}(X)=\mathsf{Can}(X')\) whenever \(X\sim X'\). The canonical recovery problem asks for \(\mathsf{Can}(X^{\star})\).
\end{definition}

This is important for cryptography because a key derivation function cannot depend on arbitrary irrelevant differences if the public observable intentionally creates equivalent witnesses. One may derive a secret from a canonical representative, from the planted witness, or from hidden perturbation data. Each choice produces a different security problem.

\begin{proposition}[Quotient identifiability]
\label{prop:quotient_identifiability}
Let \(S\subseteq\mathcal{X}\) be the generation support and let \(\sim\) be an equivalence relation on \(S\). Recovery modulo \(\sim\) from \(Y=\Phi(X)\) is information theoretically possible for every generated object if and only if
\[
\Phi(X)=\Phi(X')\quad\Longrightarrow\quad X\sim X'
\]
for all \(X,X'\in S\). Equivalently, the assignment
\[
\psi:\Phi(S)\to S/{\sim},
\qquad
\psi(y)=[X]
\]
where \(X\) is any element of \(S\) with \(\Phi(X)=y\), is well defined. If, in addition, \(\sim\) is observable compatible, then the nonempty observable fibers are exactly the equivalence classes and \(\Phi\) induces a bijection between \(S/{\sim}\) and \(\Phi(S)\).
\end{proposition}

\begin{proof}
If the displayed implication holds, then every nonempty fiber of \(\Phi\) contained in \(S\) lies inside a single equivalence class. Therefore \(\psi(y)\) is independent of the representative \(X\) chosen from the fiber over \(y\), so the public value \(Y\) determines the equivalence class of \(X\). An unbounded procedure can recover that class by exhaustive search over \(S\). Conversely, if there are \(X,X'\in S\) with \(\Phi(X)=\Phi(X')\) but \(X\not\sim X'\), then the same public observation is consistent with two different equivalence classes. No procedure using only \(Y\) and public parameters can output the correct class for both hidden objects. Finally, if observable compatibility also holds, each equivalence class is contained in one fiber, while the implication above says each fiber is contained in one equivalence class. Hence the nonempty fibers and equivalence classes coincide.
\end{proof}

\subsection{Entropy and list size requirements for cryptographic hardness}

Let \(X^{\star}\) be the secret microscopic object and let \(Y=\Phi(X^{\star})\). Entropy is a necessary diagnostic, but it is not a security proof. The most useful quantities for the present framework are the conditional support size, the size of high probability conditional lists, and the guessing probability after observing \(Y\).

\begin{definition}[Conditional support and high mass lists]
For an observed value \(y\), define
\[
S_y=\{x:\Prob[X^{\star}=x\mid Y=y]>0\}.
\]
For \(0\le \delta<1\), define the conditional list size
\[
L_\delta(y)=\min\left\{|L|:L\subseteq S_y,\ \Prob[X^{\star}\in L\mid Y=y]\ge 1-\delta\right\}.
\]
The worst case list size over the support of \(Y\) is
\[
L_\delta^{\max}=\max_{y:\Prob[Y=y]>0}L_\delta(y).
\]
\end{definition}

\begin{proposition}[What a small conditional list does and does not give]
\label{prop:small_list_caution}
Suppose that for every observed \(y\) an adversary can compute a list \(L_y\subseteq S_y\) with \(|L_y|\le 2^s\) and
\[
\Prob[X^{\star}\in L_y\mid Y=y]\ge 1-\delta.
\]
Then the following statements hold.
\begin{enumerate}[label=(\alph*)]
\item If the adversary is given an additional selector \(\mathsf{Sel}_y\) that accepts exactly the planted witness inside \(L_y\), exhaustive testing over \(L_y\) recovers \(X^{\star}\) with probability at least \(1-\delta\) using at most \(2^s\) selector calls.
\item Without such a selector, the list alone does not imply planted exact recovery. Conditional on \(X^{\star}\in L_y\), the optimal success probability using only the posterior distribution on \(L_y\) is
\[
\max_{x\in L_y}\Prob[X^{\star}=x\mid Y=y,\ X^{\star}\in L_y].
\]
\item If the goal is the arbitrary witness game and \(L_y\subseteq\Phi^{-1}(y)\), then any element of \(L_y\) is a valid witness, but this solves a different problem from planted exact recovery.
\end{enumerate}
\end{proposition}

\begin{proof}
For (a), enumerate the at most \(2^s\) elements of \(L_y\) and apply the selector. When the planted object is in the list, the selector identifies it uniquely by assumption. The stated success probability follows from the conditional list guarantee. For (b), once the only available information is \(Y=y\) and membership in \(L_y\), any recovery rule must choose some element of \(L_y\). The Bayes optimal choice is a posterior maximizer, giving exactly the displayed probability. Public verification of \(\Phi(x)=y\) cannot by itself distinguish among multiple list elements in the same fiber. For (c), the arbitrary witness game accepts any \(x\) satisfying \(\Phi(x)=y\), so any listed element is sufficient, but the planted object need not be recovered.
\end{proof}

\begin{definition}[Conditional guessing probability]
For an observed value \(y\), define
\[
p_{\max}(y)=\max_x\Prob[X^{\star}=x\mid Y=y].
\]
The worst case conditional guessing probability is
\[
p_{\max}^{\mathrm{wc}}=\max_{y:\Prob[Y=y]>0}p_{\max}(y).
\]
The corresponding worst case conditional min entropy is
\[
H_{\infty}^{\mathrm{wc}}(X^{\star}\mid Y)=-\log_2 p_{\max}^{\mathrm{wc}}.
\]
\end{definition}

\begin{proposition}[Average guessing probability]
\label{prop:average_guessing_probability}
Let \(X\) and \(Y\) be discrete random variables on finite sets. The maximum success probability of any recovery rule that observes only \(Y\) is
\[
P_{\mathrm{guess}}(X\mid Y)=\sum_y \Prob[Y=y]\max_x\Prob[X=x\mid Y=y].
\]
Equivalently,
\[
P_{\mathrm{guess}}(X\mid Y)=\mathbb{E}_Y\left[\max_x\Prob[X=x\mid Y]\right].
\]
\end{proposition}

\begin{proof}
For each observed value \(y\), the best rule outputs a posterior maximizer. Its conditional success probability is \(\max_x\Prob[X=x\mid Y=y]\). Averaging over \(y\) gives the displayed formula.
\end{proof}

\begin{proposition}[Min entropy is only a one guess bound]
If \(H_{\infty}^{\mathrm{wc}}(X^{\star}\mid Y)<s\), then for some public value \(y\) there exists a single candidate with conditional probability greater than \(2^{-s}\). This does not imply that the entire conditional support has size below \(2^s\).
\end{proposition}

\begin{proof}
The inequality \(H_{\infty}^{\mathrm{wc}}(X^{\star}\mid Y)<s\) is equivalent to \(p_{\max}^{\mathrm{wc}}>2^{-s}\). Therefore there is some observed value \(y\) and some candidate \(x\) with \(\Prob[X^{\star}=x\mid Y=y]>2^{-s}\). The statement concerns the best single guess. It says nothing by itself about the number of other candidates of smaller probability, so it cannot be used as a support size theorem.
\end{proof}

\begin{remark}[What must be measured]
For implementation research, the relevant empirical quantities are not only \(H_{\infty}\), but also estimated list sizes, collision rates, conditional candidate counts for reduced parameter sets, and the success probability of concrete attacks. Entropy can be high while a structural algorithm is efficient, and entropy can be low without giving a simple enumeration algorithm if the high probability list is hard to compute. Both failures must be tested.
\end{remark}

\begin{proposition}[Short public digest limitation]
Let \(h:\mathcal{Y}\to\{0,1\}^\ell\) be the only public value published from \(Y\). Then any recovery problem based only on \(h(Y)\) has at most \(2^\ell\) distinct public transcripts.
\end{proposition}

\begin{proof}
The image of \(h\) has cardinality at most \(2^\ell\). Therefore the public transcript partitions all microscopic objects into at most \(2^\ell\) visible classes. No claim about a larger number of distinguishable public instances can be made from the transcript alone.
\end{proof}

For this reason, the cryptographic version of the framework treats \(Y\) itself as the public key. A hash may be used as an identifier, checksum, or final key derivation step, but it must not replace the observable vector when the intended hardness depends on the structure and size of \(Y\). A 256 bit digest can be appropriate for deriving a 256 bit session key from a recovered witness; it is not appropriate as the only public object if the mathematical claim concerns a much larger exact recovery landscape.

\subsection{Structural attacks}

The main cryptographic risk is not exhaustive search. The main risk is that \(\Phi\) accidentally creates a simpler problem. A useful observable family must avoid trivial telescoping, excessive linearity, low rank structure, small local decompositions, separable halves, and smooth relaxations that make exact rounding easy. The next section states these attacks more formally and gives corresponding mitigation principles.

\begin{figure}[t]
\centering
\begin{tikzpicture}[node distance=1.4cm, every node/.style={font=\small}]
\node[draw, rounded corners, fill=blue!8, align=center, minimum width=3.0cm] (gen) {Sample hidden path\\\(X^{\star}\)};
\node[draw, rounded corners, fill=green!8, align=center, right=of gen, minimum width=3.1cm] (obs) {Observable map\\\(Y=\Phi(X^{\star})\)};
\node[draw, rounded corners, fill=orange!10, align=center, right=of obs, minimum width=2.7cm] (pub) {Public key\\\(\pk=Y\)};
\node[draw, rounded corners, fill=red!8, align=center, below=of obs, minimum width=3.6cm] (adv) {Adversary tries exact\\ recovery of \(X^{\star}\)};
\draw[-{Latex[length=2mm]}, thick] (gen) -- (obs);
\draw[-{Latex[length=2mm]}, thick] (obs) -- (pub);
\draw[-{Latex[length=2mm]}, thick] (pub) |- (adv);
\draw[-{Latex[length=2mm]}, thick, dashed] (gen) -- node[left]{secret witness} (adv);
\end{tikzpicture}
\caption{Cryptographic one way interpretation. The public object is a large observable vector, while the secret witness is the exact microscopic path object.}
\label{fig:crypto_relation}
\end{figure}
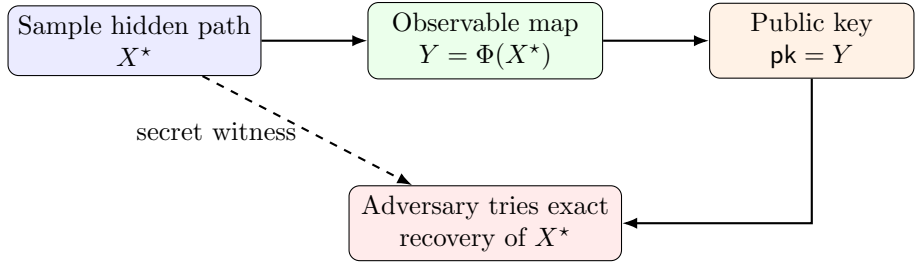

\section{Attack taxonomy and mitigation principles}

This section is deliberately hostile. The purpose is to describe how the proposed relation could fail and what design restrictions should be imposed before any cryptographic claim is attempted. None of the mitigation principles below is a proof of security. They are filters that a candidate observable family must survive.

\subsection{Linearization attacks}

The first attack is to forget the path language and ask whether the public observables are just a disguised linear system.

\begin{theorem}[Affine observable collapse]
Assume \(q\) is prime. Let the microscopic object be encoded as \(v\in\F_q^N\), and suppose that the public observable has the affine form
\[
\Phi(v)=Av+c,
\]
where \(A\in\F_q^{m\times N}\) and \(c\in\F_q^m\) are public. Let \(r=\rank(A)\). If \(r=N\), then the unique witness is recoverable in polynomial time by Gaussian elimination. If \(r<N\) and the support is the full space \(\F_q^N\), then an arbitrary witness in the fiber is recoverable in polynomial time, and the fiber has exactly \(q^{N-r}\) elements.
\end{theorem}

\begin{proof}
The equation \(\Phi(v)=Y\) is equivalent to \(Av=Y-c\). If \(r=N\), the linear map is injective, and Gaussian elimination recovers the unique solution. If \(r<N\) and \(Y-c\in\im(A)\), Gaussian elimination gives one particular solution \(v_0\) and a basis of \(\Ker(A)\). The full solution set is \(v_0+\Ker(A)\), whose size is \(|\Ker(A)|=q^{N-r}\). Therefore an arbitrary valid witness is obtained efficiently.
\end{proof}

\textbf{Mitigation.} Every proposed observable family should be tested for affine collapse over \(\F_q\), over lifted integer representations, and over low degree linearizations of its coordinates. Nonlinearity should be introduced before aggregation, not only after a large linear map has already leaked most of the structure. A practical mitigation rule is to reject any parameter set for which a linear or affine surrogate recovers either the planted witness, an arbitrary witness, or a small conditional list.

\subsection{Approximate linearity and lattice reduction}

Gaussian noise is useful only if it does not turn the system into a standard approximate linear problem with weak parameters. If the observables admit a lift of the form
\[
Y=A v+e
\]
over the integers, with \(e\) small, the instance may resemble bounded distance decoding, integer least squares, or learning with errors \cite{regev2005,micciancio2009}. This is not automatically bad, but it means that lattice reduction becomes the natural attack.

\begin{remark}[Two safe interpretations]
There are two coherent choices. One may intentionally reduce the system to a known lattice style assumption and select parameters using lattice estimators. Or one may try to avoid approximate linear structure altogether. What is dangerous is to accidentally create an LWE like instance while choosing parameters as if no lattice attack existed.
\end{remark}

\textbf{Mitigation.} For every observable coordinate, compute linear and affine fits over several lifts, estimate residual widths, and run lattice reduction experiments on reduced parameters. If the residual is small and structured, treat the construction as a lattice problem and parameterize it accordingly. If the goal is not to rely on lattices, add nonlinear local mixing, cross time couplings, and quantization rules that prevent a stable small error linear model.

\subsection{Dynamic programming attacks}

A path problem can look enormous while still being solvable if the observable decomposes locally in time.

\begin{theorem}[Dynamic programming collapse for small additive state]
Let \(\mathcal{S}\) be a finite state space, let \(a,b\in\mathcal{S}\) be fixed endpoints, and let \(\mathcal{A}\) be a finite abelian group of observable accumulators. Suppose paths are constrained by local transitions and the observable has the form
\[
\Phi(x_0,\ldots,x_T)=\sum_{i=0}^{T-1}\phi_i(x_i,x_{i+1})\in\mathcal{A}.
\]
If \(|\mathcal{S}|=V\) and \(|\mathcal{A}|=M\), then the existence of a path from \(a\) to \(b\) with observable \(Y\) can be decided in time \(O(TV^2M)\) and space \(O(TVM)\) by dynamic programming. With parent pointers, one valid path can also be recovered within the same asymptotic time.
\end{theorem}

\begin{proof}
For each time \(i\), state \(s\in\mathcal{S}\), and accumulator value \(u\in\mathcal{A}\), maintain whether there exists a partial path ending at \(s\) with accumulated observable \(u\). Initialize the table at \((0,a,0)\). For each transition \((s,t)\), update \(u\) to \(u+\phi_i(s,t)\). There are \(T\) layers, at most \(V^2\) transitions per layer, and \(M\) accumulator values. The table has \(O(TVM)\) entries, and storing a parent pointer for each reachable entry permits reconstruction of one path. At the end, check whether \((T,b,Y)\) is reachable.
\end{proof}

\textbf{Mitigation.} Avoid observables that are purely sums of local transition costs into a small accumulator. Add global couplings involving distant time indices, randomized index permutations, block interactions, and observables that depend on nonlocal patterns of the path. If a dynamic program can keep a small sufficient state, the candidate should be rejected.

\subsection{Meet in the middle attacks}

Even when a full dynamic program is too large, separability can give a square root attack.

\begin{proposition}[Meet in the middle for separable observables]
Suppose the microscopic object splits as \(X=(X_L,X_R)\) and the observable satisfies
\[
\Phi(X)=\Phi_L(X_L)+\Phi_R(X_R)
\]
in a public finite group. If the left and right search spaces have sizes \(N_L\) and \(N_R\), then an arbitrary witness can be found using \(O(N_L+N_R)\) group evaluations and \(O(N_L)\) stored table entries, assuming constant expected time table lookup and ignoring output verification cost.
\end{proposition}

\begin{proof}
Compute and store all pairs \((\Phi_L(x_L),x_L)\) for left candidates. For every right candidate \(x_R\), compute \(Y-\Phi_R(x_R)\) and look it up in the table. A match gives \(\Phi_L(x_L)+\Phi_R(x_R)=Y\). The cost is dominated by enumerating both sides and performing table lookups.
\end{proof}

\textbf{Mitigation.} Use observables that couple early and late path segments so that no clean split exists. Dense cross block terms, randomly selected long range interactions, and challenge dependent global mixing reduce the usefulness of meet in the middle enumeration. Reduced parameter experiments should explicitly test split points at many time indices.

\subsection{SAT, SMT, and algebraic solving}

If the alphabet is small and the constraints are sparse, the exact recovery problem may be easier as a constraint satisfaction problem than as a path integral inspired problem.

\begin{proposition}[Constraint encoding warning]
Any finite exact recovery instance with \(N\) variables over \(\Z_q\) and a verification circuit of size \(C\) can be encoded as SAT using \(O(N\lceil\log_2 q\rceil+C)\) Boolean variables and \(O(C+N\lceil\log_2 q\rceil)\) clauses after a standard Tseitin transformation.
\end{proposition}

\begin{proof}
Encode each \(\Z_q\) variable by \(\lceil\log_2 q\rceil\) Boolean variables and constrain unused binary representations if \(q\) is not a power of two. Translate the arithmetic circuit computing \(\Phi(X)\) and the equality test \(\Phi(X)=Y\) into a Boolean circuit of size \(C\). A standard Tseitin transformation introduces one auxiliary variable per gate and a constant number of clauses per gate, producing an equisatisfiable CNF with the claimed asymptotic size.
\end{proof}

\textbf{Mitigation.} Large dimension alone is not enough. The verification circuit must be dense and globally coupled enough that SAT, SMT, and algebraic solvers do not exploit sparsity. The software should generate small instances, solve them with off the shelf solvers, and extrapolate cautiously. If reduced instances scale too gently, the observable family is not a good candidate.

\subsection{Approximation, rounding, and saddlepoint guided attacks}

The slogan of the paper is that approximation is not exact recovery. This is true only when approximate solutions cannot be rounded into exact witnesses.

\begin{lemma}[Rounding attack condition]
Let \((\mathcal{X},d)\) be a metric space of encoded microscopic objects. Suppose valid witnesses for a given public value are separated by distance at least \(\delta>0\). If an approximation algorithm returns \(\widetilde{X}\) with \(d(\widetilde{X},X^{\star})<\delta/2\), then ideal nearest neighbor rounding over the valid witness set recovers \(X^{\star}\) uniquely.
\end{lemma}

\begin{proof}
If there were another valid witness \(X'\ne X^{\star}\) with \(d(\widetilde{X},X')<\delta/2\), then by the triangle inequality
\[
d(X^{\star},X')\le d(X^{\star},\widetilde{X})+d(\widetilde{X},X')<\delta,
\]
contradicting the separation assumption. Hence \(X^{\star}\) is the unique valid witness within radius \(\delta/2\).
\end{proof}

\textbf{Mitigation.} The design should avoid stable basins in which continuous approximation lands close enough for deterministic rounding. Micro perturbations, discrete Gaussian ambiguity, quantization, and many near macroscopic competitors can all increase the gap between coarse localization and exact recovery. This claim must be tested: if saddlepoint or gradient methods followed by rounding recover the witness, the candidate fails.

\subsection{Statistical leakage and multi instance attacks}

A single public instance may look safe while many instances reveal the generator. If macro increments, micro perturbations, projection salts, or Gaussian samples are generated from related seeds, public keys can leak correlations.

\begin{definition}[Multi instance leakage test]
A distribution family passes a basic multi instance leakage test at sample size \(N_s\) if no tested distinguisher can distinguish \(N_s\) public keys from the proposed generator versus an idealized reference distribution with non negligible advantage.
\end{definition}

\textbf{Mitigation.} Use independent domain separated randomness for each source: macro increments, micro perturbations, noise, projection choices, nonlinear mixing constants, and quantization salts. If rejection sampling is used by the path generator, publish and analyze the acceptance predicate as part of the distribution, since rejection criteria can bias the accepted paths.

\subsection{Compression and generator attacks}

A path that can be compressed is not automatically weak, but compression is a warning that the generator has structure. A path generator that mines for irregular looking outputs can also introduce selection bias.

\textbf{Mitigation.} Measure compressibility of the witness encoding, entropy of increment histograms, autocorrelation, Fourier spectra of coordinate sequences, and predictability of \(x_{i+1}\) from previous states. The generator should reject obvious low complexity paths, but it must also avoid selecting such a narrow class that the rejection predicate itself becomes exploitable.

\subsection{Witness collision attacks}

The distinction between planted recovery and arbitrary witness finding is central. If any witness in the fiber \(\Phi^{-1}(Y)\) can be used to derive the same secret or pass the same authentication, then the adversary does not need \(X^{\star}\). It only needs some \(X\) with \(\Phi(X)=Y\).

\textbf{Mitigation.} A future primitive must choose one of three strategies. First, enforce near injectivity on the generated support and treat collisions as failures. Second, define a canonical equivalence class and derive secrets only from canonical data. Third, make the planted witness contain additional secret components that alternative public witnesses do not provide. Each strategy changes the assumption and must be stated explicitly.

\subsection{Quantum attacks beyond Grover}

Grover search gives a generic square root speedup \cite{grover1996}; Shor's algorithm does not directly apply unless the observable family hides an abelian group structure suitable for a hidden subgroup formulation \cite{shor1997}. However, other quantum techniques may exploit walks on the path graph, amplitude amplification of a structural subroutine, or hidden shift like symmetries.

\textbf{Mitigation.} Parameters should include a quantum margin: an intended \(\lambda\) bit quantum security level requires at least about \(2\lambda\) bits of effective generic search space before structural attacks. More importantly, the observable family should be checked for group actions, periodicities, and symmetries that could turn the problem into a known quantum friendly structure.

\subsection{Mitigation matrix}

\begin{longtable}{p{0.20\textwidth}p{0.34\textwidth}p{0.34\textwidth}}
\toprule
Attack surface & Failure mode & Mitigation principle \\
\midrule
Linearization & Gaussian elimination recovers a witness or a small affine list & Reject low rank affine surrogates, mix nonlinearly before aggregation, test all lifted linear models \\
Approximate linearity & Lattice reduction solves a small error system & Either parameterize as a lattice assumption or destroy stable small error linear lifts \\
Dynamic programming & Local additive observables admit a small sufficient state & Add global time couplings and long range interactions \\
Meet in the middle & The observable separates into independent halves & Use cross block terms and randomized long range dependencies \\
SAT or SMT & Sparse constraints over small alphabets solve efficiently & Increase global density, test reduced instances with solvers, avoid overly local constraints \\
Approximation and rounding & Saddlepoint or gradient methods reach a roundable basin & Use micro perturbations, quantization, and competing near macroscopic paths \\
Statistical leakage & Many public keys reveal generator bias & Domain separate randomness and test distributions against distinguishers \\
Compression & Witnesses come from a low complexity subfamily & Measure compressibility and avoid exploitable rejection bias \\
Witness collision & Any fiber element breaks the primitive & Decide between injective support, canonical equivalence, or planted secret data \\
Quantum structure & Hidden group structure gives more than Grover & Search for symmetries and keep a quantum security margin \\
\bottomrule
\end{longtable}

\section{From one way relation to primitive design}

A one way relation is not yet a public key encryption scheme. This section describes what would be required to turn exact hidden paths into a primitive.

\subsection{Hashing the secret witness}

A simple one way function candidate is
\[
f_{\pp}(X)=\Phi_{\pp}(X).
\]
A derived secret may be defined as
\[
K=H(\Enc(X)),
\]
where \(H\) is a cryptographic hash and \(\Enc(X)\) is an exact encoding of the microscopic object. This is useful for commitments, key derivation from a secret witness, or challenge response, but it is not a public key encryption mechanism by itself because a public sender cannot compute \(K\) without \(X\).

\begin{proposition}[When arbitrary witnesses give the same derived value]
\label{prop:kdf_factors}
Let \(K:\mathcal{X}\to\mathcal{K}\) be a deterministic value derived from a witness. The following are equivalent.
\begin{enumerate}[label=(\alph*)]
\item \(K(X)=K(X')\) whenever \(\Phi(X)=\Phi(X')\).
\item There exists a function \(G:\Phi(\mathcal{X})\to\mathcal{K}\) such that \(K=G\circ\Phi\).
\end{enumerate}
Therefore, if a primitive is intended to accept any witness in a fiber as equivalent, the derived value must factor through the public observable or through an explicitly defined quotient. If \(K\) does not factor this way, different witnesses in the same fiber may produce different secrets.
\end{proposition}

\begin{proof}
If \(K=G\circ\Phi\), then \(\Phi(X)=\Phi(X')\) immediately implies \(K(X)=K(X')\). Conversely, suppose \(K\) is constant on every fiber of \(\Phi\). For \(y\in\Phi(\mathcal{X})\), choose any \(X_y\) with \(\Phi(X_y)=y\) and define \(G(y)=K(X_y)\). This definition is independent of the chosen representative because \(K\) is constant on fibers. Then \(G(\Phi(X))=K(X)\) for every \(X\).
\end{proof}

\subsection{Identification protocol sketch}

A natural first primitive is an identification protocol rather than encryption. The prover knows \(X^{\star}\) such that \(Y=\Phi(X^{\star})\). The verifier knows \(Y\). A protocol would ask the prover to demonstrate knowledge of a valid witness without revealing it.

A simple but insecure template is:
\[
\text{commit to a masked path},
\]
\[
\text{receive a challenge},
\]
\[
\text{open selected consistency checks}.
\]
This resembles the high level structure of many zero knowledge proof systems, but a secure protocol would require careful soundness, zero knowledge simulation, and resistance to witness extraction. The exact hidden path framework supplies a relation; it does not automatically supply a secure proof system.

\subsection{Commitment style interpretation}

One can also view \(Y=\Phi(X)\) as a commitment to a hidden path. Binding means that it is hard to find two distinct openings \(X\ne X'\) such that
\[
\Phi(X)=\Phi(X').
\]
Hiding means that \(Y\) does not reveal useful information about \(X\). These two requirements are in tension. Strong hiding usually implies large fibers, while strong binding demands small or hard to find fibers. This tension is exactly why observable design is central.

\begin{definition}[Collision resistance of the observable map]
An observable family is collision resistant on a distribution \(\mathcal{X}_\lambda\) if no probabilistic polynomial time adversary, given the public parameters for security parameter \(\lambda\), can output \(X\ne X'\) in the admissible domain with
\[
\Phi(X)=\Phi(X')
\]
with non negligible probability in \(\lambda\).
\end{definition}

Collision resistance is stronger than planted recovery hardness in some directions and weaker in others. It rules out finding any two colliding witnesses, but it does not necessarily imply that a planted witness cannot be recovered from its image.

\subsection{Toward key encapsulation}

A future key encapsulation mechanism would need an additional public operation that lets an encapsulator create a ciphertext from \(Y\) while allowing only the holder of \(X^{\star}\) to derive the shared secret. The most conservative statement is therefore:

\begin{quote}
Exact hidden path recovery supplies a candidate hard relation. A KEM requires an additional trapdoor transformation, reconciliation mechanism, or proof based compiler that has not yet been specified in this paper.
\end{quote}

This prevents a common mistake: confusing a one way relation with public key encryption. The present paper belongs at the assumption and relation level. The later software project can generate instances, measure observable leakage, and test whether such a relation is worth building on.

\section{Cryptographic parameter considerations}

Although concrete parameters are premature, the framework suggests several measurable quantities.

\subsection{Public key size}

Let the observable vector contain \(m\) entries, each carrying \(\ell\) bits after quantization. Then the public observable size is
\[
L=m\ell.
\]
Unlike many deployed schemes, the present framework may naturally require large public keys. This is not automatically a flaw at the research stage. A large public key can be mathematically necessary if the goal is to publish enough observable structure for verification while avoiding collapse into a short digest.

If the generated support has size \(|S_\lambda|\) and exact injective recovery on that support is desired, then Corollary~\ref{cor:digest} gives the necessary condition
\[
L\ge \lceil\log_2 |S_\lambda|\rceil.
\]
If the target is not injective recovery but relation finding, then this lower bound is no longer the right requirement; instead one must analyze the number and usefulness of alternative witnesses in each fiber.

\begin{example}[Observable size]
If \(m=4096\) and each observable entry has \(16\) bits, then
\[
L=65536\text{ bits}=8192\text{ bytes}.
\]
This is a large but not absurd public object for experimental cryptographic research.
\end{example}

\subsection{Effective security bits}

A conservative generic estimate should be based on guessing probability rather than on the raw size of the hidden path space. After observing \(\pp\) and \(Y\), define the average optimal guessing probability
\[
P_{\mathrm{guess}}=
\mathbb{E}_{\pp,Y}\left[\max_x \Prob[X^{\star}=x\mid \pp,Y]\right].
\]
The associated guessing entropy is
\[
H_{\mathrm{guess}}(X^{\star}\mid \pp,Y)=-\log_2 P_{\mathrm{guess}}.
\]
For a distribution with nearly flat conditional fibers, this quantity is close to the logarithm of the typical conditional fiber size. In that idealized case, a generic classical exhaustive search has scale roughly
\[
2^{H_{\mathrm{guess}}(X^{\star}\mid \pp,Y)},
\]
and a generic Grover style quantum search has scale roughly
\[
2^{H_{\mathrm{guess}}(X^{\star}\mid \pp,Y)/2}.
\]
These are not hardness proofs. They are only baseline estimates after the observable family has survived structural attacks. If the observable map linearizes, admits a small dynamic program, produces a tractable lattice instance, or has exploitable SAT structure, then the guessing entropy estimate may be irrelevant.

\subsection{Noise width and correctness}

Internal noise can increase microscopic ambiguity, but cryptographic correctness depends on what must be recovered. If the holder of the secret stores \(\eta^{\star}\), then exact verification can include it. If the noise is not stored, a later protocol must tolerate errors or use reconciliation. These are different designs.

\begin{proposition}[Stored noise exactness]
If \(\eta^{\star}\) is sampled during key generation and stored as part of \(\sk\), then exact recomputation of
\[
Y=\Phi(X^{\star})=\Phi(x_0^{\star},\Delta^{\star},\eps^{\star},\eta^{\star})
\]
is deterministic.
\end{proposition}

\begin{proof}
Once \(\eta^{\star}\) is fixed and stored, it is no longer random during verification. The observable map receives the same complete microscopic input object and therefore returns the same output.
\end{proof}

\subsection{Distributional independence}

If many public keys are generated with related randomness, the path distribution may leak. A cryptographic generator must therefore domain separate all random sources used for macro increments, micro perturbations, Gaussian noise, projection choices, and quantization salts. In mathematical notation this means that the components of \(\pp\) and \(X^{\star}\) should not share hidden low dimensional seeds unless the induced correlations are analyzed.

\section{Attack checklist for future implementations}

A future implementation should not only generate visually irregular high dimensional paths. It must attempt to break them. The following checklist is part of the proposed research methodology for the future PathFrog software and for the path generation module tentatively called QuantumMiner.

\begin{enumerate}[label=(\roman*)]
\item Run exact enumeration for tiny parameters and compute true fiber sizes.
\item Run Gaussian elimination on direct, lifted, projected, and approximate linear forms of the observables.
\item Run lattice reduction when an approximate linear model with small residuals is found.
\item Encode small and medium instances as SAT, SMT, and polynomial systems.
\item Test dynamic programming attacks for every decomposition over time, blocks, or coordinates.
\item Test meet in the middle splits at many time indices and coordinate partitions.
\item Apply continuous relaxations, saddlepoint localization, gradient methods, and rounding.
\item Train simple predictors to detect distributional regularities in generated paths.
\item Measure compression of secret paths and reject highly compressible samples.
\item Estimate conditional collision rates by sampling alternative witnesses.
\item Compare approximate reconstruction quality with exact recovery success.
\item Run multi instance distinguishers against public key distributions.
\item Check for symmetry, periodicity, and hidden group actions.
\item Estimate generic quantum cost by halving the effective search exponent after all structural reductions.
\item Document every failed parameter set, since failed parameters are evidence about the boundary of the framework.
\end{enumerate}

A mined path is not good merely because it looks irregular. It is good only if known recovery methods fail, if the generated distribution remains well defined, and if the public observables have enough information for verification without exposing an efficient inverse map.

\section{Proof obligations and what is not claimed}

Because the framework is intended to become a cryptographic research direction, it is useful to state the proof obligations precisely. The elementary theorems in this paper prove counting, identifiability, projection, entropy, and attack collapse statements. They do not prove asymptotic hardness of any concrete observable family. A future hardness claim would require at least the following additional items.

\begin{enumerate}[label=(\roman*)]
\item A fixed asymptotic instance distribution \(\mathcal{I}_\lambda\), not only an informal generator.
\item A public observable family \(\Phi_\lambda\) with efficiently computable verification.
\item A clear target relation: planted recovery, arbitrary witness finding, canonical recovery, or recovery modulo equivalence.
\item Proof or evidence that known polynomial time reductions do not collapse the relation to linear algebra, lattice decoding with weak parameters, dynamic programming, meet in the middle, or sparse constraint solving.
\item Parameter estimates after structural reductions, not before them.
\end{enumerate}

Thus the paper proves that several common shortcuts are invalid, for example replacing a large observable vector by a short digest or treating approximate localization as exact recovery. It does not prove that every proposed observable is hard. This distinction is part of the mathematical formulation.

\section{Cryptographic status of the proposal}

The correct status of the proposal is the following.

\begin{enumerate}[label=(\alph*)]
\item It defines a candidate exact recovery relation.
\item It motivates a post quantum research direction because Shor type attacks do not directly apply.
\item It does not yet define a secure KEM, PKE, signature, or commitment scheme.
\item It requires extensive cryptanalysis of observable families, including the attack classes listed above.
\item It requires explicit mitigation choices for linearization, lattice reduction, dynamic programming, meet in the middle, SAT solving, statistical leakage, approximation plus rounding, and witness collisions.
\item It requires careful parameter selection and implementation testing.
\end{enumerate}

This status is not a weakness. Many cryptographic research programs begin by isolating a relation or assumption before producing efficient primitives. The important requirement is honesty about what has and has not been proved.

\section{Discussion}

The framework is built around a simple but important separation. Path integral methods often tolerate approximation because they seek global quantities. Exact hidden path recovery does not. If the hidden object is a discrete microscopic path, then recovering a nearby or dominant trajectory is not the same as recovering the exact object.

The most important design lessons are:

\begin{enumerate}[label=(\alph*)]
\item Large internal path spaces are easy to define, but do not by themselves prove hardness.
\item Public observables must not be compressed into a short digest if the mathematical goal is to study exact recovery.
\item High dimension and projection create ambiguity, but observable structure determines whether recovery is possible.
\item Micro perturbations and internal Gaussian noise increase the gap between coarse approximation and exact microscopic reconstruction.
\item Linear observables are analyzable but may leak too much or reduce the problem to known algebra.
\item Quantization hides precision but creates collisions.
\end{enumerate}

This paper should be read as a foundation for both mathematics and cryptographic research. It defines a family of exact recovery problems, gives elementary bounds, identifies the main structural hazards, states a candidate one way relation, and now gives an explicit attack and mitigation taxonomy. The advanced work lies in constructing observable systems that are neither trivially invertible nor information theoretically hopeless, then testing them against classical and quantum relevant attack models.

\section{Conclusion}

We introduced Exact Noisy Hidden Path Recovery as a mathematical framework for studying exact reconstruction of micro perturbed paths in noisy high dimensional discrete path spaces. Inspired by path integral thinking, the framework separates approximation of global path sum behavior from exact recovery of one microscopic trajectory.

We defined discrete path spaces over \(\Z_q^n\), transition models with macro increments, micro perturbations, and internal noise, and observable systems based on projection, aggregation, nonlinearity, and quantization. We proved elementary growth, projection, fiber, and information bounds, including the no short digest principle: the effective recovery problem is bounded by the information content of the published observables. A huge hidden path space should not be collapsed into a small hash if the objective is to study mathematical recovery.

The paper does not claim a deployed cryptographic construction. Its purpose is to establish a precise mathematical and cryptographic language for future investigation. The next stage is to build concrete observable families, analyze their fibers, test recovery algorithms, and determine whether there exist parameter regimes where coarse approximation is feasible but exact microscopic recovery remains difficult. Only after such evidence exists should one attempt a complete key encapsulation mechanism, encryption scheme, signature scheme, or proof system.

\appendix

\section{Notation summary}

\begin{longtable}{ll}
\toprule
symbol & meaning \\
\midrule
\(q\) & modulus or finite field size \\
\(n\) & state dimension \\
\(T\) & path length \\
\(x_i\) & state at time \(i\) \\
\(\gamma\) & state path \((x_0,\ldots,x_T)\) \\
\(\Delta_i\) & macro increment \\
\(\eps_i\) & micro perturbation \\
\(\eta_i\) & internal noise \\
\(\mathcal{D}\) & allowed macro increment set \\
\(\mathcal{E}\) & allowed micro perturbation set \\
\(D_\sigma\) & noise distribution, often discrete Gaussian \\
\(X\) & complete microscopic path object \\
\(\Phi\) & observable map \\
\(Y\) & public observable vector \\
\(\Enc\) & exact encoding function \\
\bottomrule
\end{longtable}

\section{Additional proof details}

\begin{lemma}[Pigeonhole form]
Let \(f:A\to B\) be a map of finite sets. If \(|A|>|B|M\), then some fiber of \(f\) has size greater than \(M\).
\end{lemma}

\begin{proof}
If every fiber had size at most \(M\), then
\[
|A|=\sum_{b\in B}|f^{-1}(b)|\le |B|M,
\]
contradicting \(|A|>|B|M\).
\end{proof}

\begin{lemma}[Endpoint constraint count for full paths]
Let \(T\ge 1\). The number of paths \((x_0,\ldots,x_T)\in(\Z_q^n)^{T+1}\) with fixed endpoints \(x_0=a\) and \(x_T=b\) is \(q^{n(T-1)}\).
\end{lemma}

\begin{proof}
The intermediate states \(x_1,\ldots,x_{T-1}\) are arbitrary elements of \(\Z_q^n\). There are \(q^n\) choices for each of the \(T-1\) intermediate states.
\end{proof}

\section{A possible research program}

A natural continuation of this work is the construction of explicit observable families with tunable information content. One possible program is:

\begin{enumerate}[label=(\roman*)]
\item Define several observable families, including linear, quadratic, local nonlinear, quantized, and mixed systems.
\item For small parameters, compute exact fiber sizes by enumeration.
\item For medium parameters, estimate collision rates by sampling.
\item Compare recovery algorithms across the hierarchy of recovery notions.
\item Identify observables that avoid telescoping and obvious linearization.
\item Study whether exact recovery appears to scale differently from coarse recovery.
\end{enumerate}

Such a program would help determine whether the framework contains hard inverse problems or whether typical observable choices collapse under known methods. Either outcome is mathematically useful.

\end{document}